\documentclass[11pt]{article}
\usepackage[hmargin=2.5cm,vmargin=2.2cm]{geometry}
\usepackage[english]{babel}
\usepackage{amsmath,amssymb,amsthm}
\usepackage{mathrsfs}
\usepackage[pdfencoding=auto, psdextra]{hyperref}
\hypersetup{colorlinks=true,allcolors=[rgb]{0,0,0.6}}
\usepackage{enumitem}
\usepackage{appendix}
\usepackage[utf8]{inputenc}

\theoremstyle{plain}
\newtheorem{thm}{Theorem}[section]
\newtheorem{assu}{Assumption}
\newtheorem{cor}[thm]{Corollary}
\newtheorem{prop}[thm]{Proposition}
\newtheorem{lm}[thm]{Lemma}
\newtheorem{dfn}[thm]{Definition}

\newtheorem{remark}[thm]{Remark}

\usepackage{color}

\newenvironment{sproof}{%
  \proof}{\endproof}

\title{On the general principle of the mean-field approximation for many-boson dynamics}
\author{Clément Rouffort}

\begin{document}


\def\N{\mathbb{N}}
\def\Z{\mathscr{Z}}
\def\R{\mathbb{R}}
\def\C{\mathbb{C}}
\def\P{\mathbb{P}}
\def\Tr{\mathrm{Tr}}
\def\Id{\mathrm{Id}}
\def\Im{\mathrm{Im}}
\def\Re{\mathrm{Re}}
\def\b{\overset{\sim}{b}}
\newcommand{\zed}{\mathcal{Z}_{0}}
\newcommand{\zeds}{\mathcal{Z}_{s}}
\newcommand{\zedsi}{\mathcal{Z}_{\sigma}}

\maketitle

\begin{abstract}
The mean-field approximations of  many-boson dynamics are known to be  effective in many physical relevant situations.
The  mathematical justifications of  such approximations rely generally  on  specific considerations which depend too much on the model and  on the initial states of the system which are required to be well-prepared.  In this article, using the method of Wigner measures, we prove  in a fairly complete generality  the accuracy of the mean-field approximation. Roughly speaking, we show that
the dynamics of a  many-boson system are well approximated, in the limit of a large number of particles,  by a one particle mean-field equation if the following general principles are satisfied:
\begin{itemize}
\item The Hamiltonian is in a mean-field regime (i.e.: The interaction and the free energy parts are of the same order with respect to the number of particles).
\item The interaction  is relatively compact with respect to a one particle and it  is dominated by the free energy part.
\item  There exists at most one  weak solution for the  mean-field equation for each initial condition.
\end{itemize}
The convergence towards  the mean field limit is described in terms of Wigner (probability) measures and it holds  for any initial  quantum states with a finite free energy average.  The main novelty of this article lies  in the use of fine properties of uniqueness  for the  Liouville equations in infinite dimensional spaces.
\end{abstract}

\textit{Keywords:} Mean-field theory, second quantization, transport properties, Wigner measures, nonlinear PDEs, Schr\"odinger and Hartree equations.

\section{Introduction}

The mean field theory provides effective and simple  approximations for complex systems composed of a large number of  interacting  particles or components. Generally speaking the theory is based on a principle of
averaging  of  the interaction effects exerted  on  each single particle by the others ones. From a historical point of view, the mean field theory  has its origins in the statistical mechanics of the early twentieth  century. At present, it has widespread applications in cross-cutting themes as diverse as nuclear physics,  neuroscience, and artificial intelligence (see e.g.~\cite{Negele:1982aa,MR0526368,zbMATH01149420}).

The mathematical foundations of the mean field theory are also so diverse and touching  various domains of mathematics like stochastic processes, game theory, variational calculus and nonlinear analysis
(see e.g.~\cite{MR2762362,MR1673235,MR3360742,MR3719544,MR3304746,MR3748494,MR3475664}). In particular, there  are  a large number of methods, techniques and results related to this subject (see e.g.~\cite{MR3170216,MR2313859,MR3737034,MR3695801,MR3394620,MR3461029,MR3427937}).  Therefore, it is  worth trying to unify  some of the aspects of this theory and to put them in a clear-cut form. In fact, one of the interesting questions is the validity of the mean field approximation which in principle  should emerge  as a simple consequence of symmetries and scales considerations without further specification of a model or a well prepared state of  many-body systems.

The main purpose of the present article is to prove, under general considerations, the effectiveness of the mean field approximation for quantum many-body  systems composed of identical particles. Such a many-body problem is usually described by a Hamiltonian taking the form
\begin{equation}
\label{2_eq.intro.1}
\mathbb{H}_N=\sum_{j=1}^N T_j+\frac{1}{N} \sum_{1\leq j,k\leq N} V_{j,k}\,,
\end{equation}
 where $T_j$ is interpreted as  the free energy of the $j$th particle and $V_{j,k}$ refers to the interaction between the $j$-th and $k$-th particles. For simplicity, the multiparticles interactions involving more than three bodies are not considered here.  Generally speaking, the mean field theory is used to approximate either the dynamics or the equilibrium and ground states of  many-body systems. The analysis is mostly made  through the  study of  the correlation functions of  time-independent  or  time-dependent  states. This article deals precisely with the second aspect and shows that the mean field approximation of the
Hamiltonian dynamics generated by \eqref{2_eq.intro.1} relies  conceptually in three principles:
\begin{itemize}
  \item Symmetry: The states describing the system \eqref{2_eq.intro.1} are invariant  with respect to any exchange of particles.   More precisely, $\mathbb{H}_N$ is considered as an operator (or a quadratic form) acting on a symmetric tensor product $\otimes^N_s\mathfrak{H}$ of a given Hilbert space $\mathfrak{H}$.  This means that one deals here with bosonic particles. The fermionic case is somewhat different and will be considered elsewhere.
  \item Mean-field scale: The interaction part in \eqref{2_eq.intro.1} is formally  of order $N$
  as the free energy part. The factor $1/N$ in front of the interaction ensures that the two parts are of the same order.
  \item Weak interaction:  The interaction effects are dominated by the main flow of the free energy
   thanks to some relative compactness properties of the interaction.   Such condition will be made precise in the subsequent section (see assumptions  \eqref{2_D1} and \eqref{2_D2}).
\end{itemize}
As it will be more clear later on, it is reasonable to believe that the above requirements are the basic conditions under which the dynamics of \eqref{2_eq.intro.1} can be approximated, as $N\to\infty$, by a one-particle mean field equation taking the form
\begin{equation}
\label{2_eq.intro.2}
i\partial_t u= T u+V_{av}(u)\,,
\end{equation}
where $V_{av}$ is a nonlinear vector field obtained as an average of the interaction part in $\mathbb{H}_N$. So, the aim here is to prove rigourously such generic principle in the basis of the symmetry and scale considerations mentioned above, no matter which specific model and specific states are considered.

\bigskip
\emph{Overview of the literature:} The mathematical study of the mean field theory for quantum systems dates back to the works of Hepp \cite{MR0332046} and Ginibre-Velo \cite{MR539736,MR530915}. In these pioneering  articles the authors rely on coherent structures. In the physical literature however such concepts are known well before (see for instance \cite{MR0223148}). Stimulated by the work of Spohn \cite{MR578142} who adapted the BBGKY approach of classical statistical mechanics to the framework of quantum systems, several authors studied the mean field approximation of many-body  Schr\"odinger operators  by means of the same method, see e.g.~\cite{MR1869286,MR1926667}. Subsequently, in a series of articles by Erd\H{o}s-Schlein-Yau  the mean field theory is used for the  derivation of non-linear schr\"odinger and Gross-Pitaevskii equations which are related to the celebrate phenomenon of  Bose-Einstein condensation \cite{MR2257859,MR2276262,MR2525781,MR2680421}. Following these works, the BBGKY approach became quite popular and the subject attracted
a much growing  attention specially from the mathematical physics and nonlinear PDEs communities
see e.g.~\cite{MR3210237,MR3385343,MR3165917,MR3246038,MR3551830,MR3500833,MR3506807}. Other approaches
were also elaborated   and others aspects studied thoroughly, see e.g.~ \cite{MR3360742,MR3293448,MR3013052,MR3395127,MR3466843,MR3170216,MR3310520,MR3743758}.

\bigskip
\emph{Overview of the main result:}
The effectiveness of the mean field approximation is proved by means of the Wigner measures method elaborated in
the series of papers \cite{MR3379490,MR2802894,MR2513969,MR2465733}.  Such method is based on two steps, namely a convergence and a uniqueness arguments similar to the BBGKY approach. The main difference lies in the
main quantities that are analyzed. While the BBGKY method uses reduced density matrices the Wigner measures method uses Fourier-Wigner transform of the time dependent  quantum states. The relationship between the two methods
is recently studied in \cite{Ammari:2018aa} where  the advantages of the latter approach is highlighted.
The method of Wigner measures was  used for the study of Schr\"odinger many-boson dynamics with singular  potentials of  Coulomb type in \cite{MR3379490}.  Subsequently, the latter result was improved in \cite{MR3661404} where a quite general framework is presented. Here, we build upon the work of Q.~Liard and prove that the mean field approximation actually relies only  in  some elementary principles.  In particular, we use for the convergence step the same assumptions (A1)-(A2) and (D1) or (D2) as in Liard's work and for the uniqueness step we get rid of  cumbersome assumptions in his result.

\section{Preliminaries and main results}
\label{2_prmres.sec}

We introduce  in the following paragraph  an abstract framework suitable for the study of many-boson dynamics and define some useful  notations. The main results of this article are stated thereafter. Several relevant examples are discussed in the subsequent section.

\medskip
The following notations are often used. If $K$  is an operator on a Hilbert space  then  the domain of $K$ is denoted by $D(K)$. If additionally $K$ is self-adjoint then the form domain of $K$, which is equal to $D(|K|^{\frac{1}{2}})$, is denoted by $Q(K)$.
The space of trace-class operators from a Hilbert space $\mathfrak{H}_1$ into another one $\mathfrak{H}_2$ is denoted by $\mathscr{L}^1(\mathfrak{H}_1,\mathfrak{H}_2)$  and the one of compact operators is denoted by $\mathscr{L}^\infty(\mathfrak{H}_1,\mathfrak{H}_2)$. If $\mathfrak{H}_1=\mathfrak{H}_2$, we  denote these spaces   by $\mathscr{L}^1(\mathfrak{H}_1)$ and $\mathscr{L}^\infty(\mathfrak{H}_1)$ respectively.

\medskip
Let $X$ a Hausdorff  topological space then $\mathscr{B}(X)$ is the Borel $\sigma$-algbera on $X$.
The set of all Borel probability measures on the  measurable space $(X, \mathscr{B}(X))$ will be denoted by
$\mathfrak{P}(X)$. If $X,Y$ are two Hausdorff topological spaces, $T:X\to Y$ is a Borel  map  and $\mu\in\mathfrak{P}(X)$ then the push-forward (or the image ) measure of $\mu$ by the map $T$ is a Borel probability measure on $Y$,
denoted by  $  T_\sharp \mu\in\mathfrak{P}(Y)$, and defined as
$$
\forall \mathcal{A}\in \mathscr{B}(Y), \qquad  T_\sharp \mu(\mathcal{A})=\mu\big(T^{-1}(\mathcal{A})\big)\,.
$$

\bigskip
\textit{Many-bosons Hamiltonian:} Let $\Z$ be a infinite dimensional separable complex Hilbert space. We denote $\langle .,. \rangle_{\Z}$ the scalar product on $\Z$ and we consider it to be anti-linear on the left part. The Fock space associated to $\Z$ is the following  direct sum of Hilbert tensor products,
\[
\Gamma(\Z) = \bigoplus_{n =0}^{\infty} \Z^{\otimes n}\,,
 \]
with the convention $ \Z^{\otimes 0}=\C$.
 For  any permutation $\sigma \in \mathfrak{S}_{n}$ of $n$ elements, one defines an unitary operator $\Pi_{\sigma} : \Z^{\otimes n} \longrightarrow \Z^{\otimes n} $ satisfying the  identity,
\[
\Pi_{\sigma} ( f_{1} \otimes ... \otimes f_{n} ) = f_{\sigma(1)} \otimes ... \otimes f_{\sigma(n)} \,,
\]
for any  $(f_{1},...,f_{n}) \in \Z^{n}$. Hence, using these operators one  defines an  orthogonal projection,
\begin{equation}
\label{2_eq.projsym}
P_{n} = \frac{1}{n!} \sum_{ \sigma \in \mathfrak{S}_{n}} \Pi_{\sigma}\,,
\end{equation}
which is a projection on the closed subspace of the $n$-fold symmetric tensor product $P_{n}  \Z^{\otimes n}$ denoted in all the  sequel by
$$
\bigvee^{n} \Z=P_{n}  \Z^{\otimes n}.
$$
 The bosonic Fock space associated to $\Z$ is the Hilbert space defined by,
\[
\Gamma_{s}(\Z) = \bigoplus_{n=0}^{\infty} \bigvee^{n} \Z ~.
\]
The dynamics of many-boson systems are given through operators acting on the above bosonic Fock spaces. Furthermore, the states of many-boson systems are convex combinations of normalized vector (or pure) states in $\Gamma_{s}(\Z)$. The use of Fock spaces is   natural for the study of
many-boson systems and furthermore  the second quantization techniques are quite useful for the mean-field theory (see e.g.~\cite{MR2465733}).

\bigskip
Consider from now on  an operator $A$ on $\Z$ and suppose for all the sequel that:
\begin{assu}
\begin{equation}
\label{2_A1}\tag{A1}
A \text{ is a non negative self-adjoint operator on } \Z.
\end{equation}
\end{assu}

\noindent
The free Hamiltonian without interaction of a many-boson system associated to the one-particle operator $A$ is
\[
H_{N}^{0} = \sum_{i=1}^{N} A_{i}\,.
\]
Here $N$ refers to the number of particles and  $A_i$ to the operators acting on $\Z^{\otimes N}$ and  defined by
\[
 A_{i} = 1^{\otimes (i-1)} \otimes A \otimes 1^{\otimes (N-i)}\,,
 \]
where $A$ in the right hand side acts in the $i^{th}$ component of the tensor product. It is known that $H_{N}^{0}$ is a non-negative self-adjoint operator on $\bigvee^{N} \Z$. A typical example for the couple $(\Z, A)$ is $\Z=L^{2}(\mathbb{R}^{d})$ and $A=-\Delta_x$. In this case, $\bigvee^{N} \Z$ is the space $L^{2}_{s}(\mathbb{R}^{dN})$ of symmetric square integrable functions  and  the operator $A_{i}$ is the Laplacian operator $- \Delta_{x_{i}}$  in the $x_i$ variable.

\medskip
In order to define a two particles interaction in a general abstract setting, we consider a symmetric quadratic form $q$ on $Q(A_{1}+A_{2}) \subset  \Z^{\otimes 2}$. Note that $A_{1}+A_{2}$ is considered as an operator on $\Z^{\otimes 2}$ verifying in particular,
\[
\forall(z_{1},z_{2}) \in D(A)^{2} ~,~ (A_{1}+A_{2})(z_{1} \otimes z_{2}) = (Az_{1})\otimes z_{2} + z_{1} \otimes (Az_{2})\,.
\]
Remark also that the subspace $Q(A_{1}+A_{2})$ contains non-symmetric vectors. In this article we consider the general following assumption \eqref{2_A2} satisfied by $q$ :

\begin{assu}
$q$ is a symmetric sesquilinear form on $Q(A_{1}+A_{2})$ satisfying:
\begin{equation}
\label{2_A2}\tag{A2}
 \exists 0 < a <1 ~,~ b > 0 ~,~ \forall u \in Q(A_{1}+A_{2}) ~,~ |q(u,u)| \leq a\, \langle u , (A_{1}+A_{2})u \rangle_{\otimes^{2} \mathcal{Z}} + b \,||u||_{\otimes^{2} \mathcal{Z}}^{2}\,.
 \end{equation}
\end{assu}

The assumptions \eqref{2_A1}-\eqref{2_A2} allow to consider the quadratic form $q$ as a bounded operator $\widetilde{q}$ acting from $Q(A_{1}+A_{2})$ equipped with the graph norm into its dual $Q^{'}(A_{1}+A_{2})$. Recall that the graph norm  on $Q(A_{1}+A_{2})$ is given by:
\[
||u||_{Q(A_{1}+A_{2})}^{2} = \langle u, (A_{1}+A_{2}+Id)\, u \rangle_{\Z^{\otimes 2 }}\,.
\]
Then, one can write:
\begin{equation}
\label{2_eq.qtilde}
 \forall(u,v) \in Q(A_{1}+A_{2}) ~,\quad q(u,v) = \langle u, \widetilde{q}\,v \rangle_{\Z^{\otimes 2 }} \,,
\end{equation}
where  the right hand side is actually a  duality bracket extending the inner product of $\Z^{\otimes 2}$.
Starting from the sesquilinear form $q$ one can construct a two-body interaction. In fact, define for any $p \in \mathbb{N}$ and $1 \leq i,j \leq p$,
\begin{equation}
\label{2_qformij}
q_{i,j}^{(p)}(\phi_{1}\otimes...\otimes \phi_{p}, \psi_{1} \otimes ... \otimes \psi_{p}) = q(\phi_{i}\otimes\phi_{j},\psi_{i}\otimes\psi_{j}) \prod_{k \neq i,j} \langle \phi_{k},\psi_{k} \rangle
\end{equation}
for any $(\phi_{1},...,\phi_{p},\psi_{1},...,\psi_{p}) \in Q(A)^{2p}$. Then as shown in \cite[Lemma 3.1]{MR3661404}, one can extend each $q_{i,j}^{(p)}$ to a unique symmetric quadratic form on $Q(H_{N}^{0})$. Each form $q_{i,j}^{(p)}$ represents the two-body interaction between the $i$-th and $j$-th particles. Such abstract construction covers all the known relevant examples of interactions.

\bigskip

The many-boson Hamiltonian in the mean field regime  is considered as the symmetric quadratic form,
 \begin{equation}
\label{2_HN}
 H_{N} = \sum_{i=1}^{N} A_{i} + \frac{1}{N} \sum_{1 \leq i,j \leq N} q_{i,j}^{(N)}\,.
 \end{equation}
   Under the assumptions \eqref{2_A1}-\eqref{2_A2}, it is shown in  \cite[Proposition 3.4 ]{MR3661404} that there exists a unique self-adjoint operator denoted also by $H_{N}$ and associated to the quadratic form \eqref{2_HN}. Moreover,
   the many-boson Hamiltonian $H_{N}$ satisfies $Q(H_{N})=Q(H_{N}^{0})$. In fact, the interaction part is
relatively form bounded with respect to $H_N^0$ and hence the self-adjointness follows by  the standard
KLMN theorem. This means in some sense that the many-boson Hamiltonian is a perturbation of the free Hamiltonian and the
interaction is dominated by $H_N^0$. Such requirement for the many-body systems ensures the existence and uniqueness of dynamics  and it is very standard  in quantum mechanics.

As explained in the introduction, the Wigner measures method consists in two steps, namely convergence and uniqueness. The convergence part relies in a principle reflecting a sort of weak interaction  expressed by a relative compactness of the interaction with respect to the free Hamiltonian in a one particle variable. In fact, we have two distinguished cases where  in one hand the compactness comes from the free part (i.e., $A$ has a compact resolvent) and in the other hand  the interaction is relatively compact in one variable. Below these requirements are  given precisely. Notice that we only need one of the assumptions to be satisfied.
\begin{assu}
$A$ has a compact resolvent and there exists a subspace $D$ dense in $Q(A)$ such that for any $\xi \in D$,
\begin{equation}
 \label{2_D1}\tag{A3}
 \begin{aligned}
 \lim\limits_{ \lambda \to + \infty} || \langle \xi | \otimes (A+1)^{-\frac{1}{2}} \,P_{2}\, \widetilde{q}\,(A_{1}+A_{2}+ \lambda)^{-\frac{1}{2}}||_{\mathscr{L}(\bigvee^{2} \Z,\Z)} = 0\,,
 \\
 \lim\limits_{ \lambda \to + \infty} || \langle \xi | \otimes (A+\lambda)^{-\frac{1}{2}}\, P_{2}\, \widetilde{q}\,(A_{1}+A_{2}+ 1)^{-\frac{1}{2}}||_{\mathscr{L}(\bigvee^{2} \Z,\Z)} = 0\,.
\end{aligned}
\end{equation}
\end{assu}
\noindent
Recall that $P_2$ is the orthogonal projection on the symmetric tensor product $\bigvee^2\Z$ given in \eqref{2_eq.projsym} and $\tilde{q}$ is related to the quadratic form $q$ according to \eqref{2_eq.qtilde}.
\begin{assu}
There exists a subspace $D$ dense in $Q(A)$ such that for any $\xi \in D$,
\begin{equation}
\label{2_D2}\tag{A4}
\langle \xi | \otimes (A+1)^{-\frac{1}{2}}\, P_{2} \,\widetilde{q}\,(A_{1}+
A_{2}+ 1)^{-\frac{1}{2}} \in \mathscr{L}^{\infty}\big(\bigvee^{2} \Z,\Z\big)\,.
\end{equation}
\end{assu}

\bigskip
\emph{Mean-field equation:}
The mean-field approximation reduces the complicate dynamics of  many-boson systems  to simpler one particle evolution equations called the mean-field equations.  Depending in the different choices of the quadratic form $q$, one obtains different mean-field equations. For instance, taking the form $q$ to be  a two-body delta interaction in one dimension defined by
\[
q(z^{\otimes2},z^{\otimes2})= \langle z^{\otimes2}, \lambda \,\delta(x_{1}-x_{2}) z^{\otimes2} \rangle = \lambda ||z||_{L^{4}(\mathbb{R})}^{4} \,,
\]
then the mean field equation in this case is the cubic nonlinear Schrödinger equation,
\begin{equation}
\label{2_eq.nlsh1}
i\partial_{t}\varphi=-\Delta \varphi+\lambda |\varphi|^2 \varphi\,.
\end{equation}
In this case, the self-adjoint operator $A$ corresponds to the operator $-\Delta_{x}$. Note that several other examples  are provided  in Section \ref{2_sec.examples}.
 
 We give below an abstract way of defining the correct mean-field equation  corresponding to the many-boson Hamiltonian \eqref{2_HN}.
Let $(Q(A),||.||_{Q(A)})$ be the form domain of the operator $A$ equipped with the graph norm,
\[
||u||_{Q(A)}^{2} = \langle u,(A+1)u \rangle ~,~ u \in Q(A),
\]
and $Q^{'}(A)$ its dual with respect to the inner product of $\Z$.  Since these spaces will be used throughout all the article, we make  the following shorthand notations for such Hilbert rigging:
\begin{equation}
\Z_1=(Q(A),||.||_{Q(A)}) \subset \Z \subset \Z_{-1}=(Q'(A),||.||_{Q'(A)})\,.
\end{equation}
Using the quadratic form $q$, we can define a quartic monomial $q_{0}$ associated to $q$,
\begin{equation}
\label{2_intr.q0}
\forall z \in \Z_1 ~,~ q_{0}(z) = \frac{1}{2}q(z\otimes z, z \otimes z) \,.
\end{equation}
The particularity of this monomial is that it is Gâteaux differentiable on $\Z_1$ and the map $u \longrightarrow \partial_{\overline{z}} q_{0}(z)[u]$ is an anti-linear continuous form on $Q(A)$ and hence $ \partial_{\overline{z}} q_{0}(z)$ can be identified with a vector belonging to $\Z_{-1}$ by the Riesz theorem. We can then define the mean-field equation as,
\begin{equation}
\quad \quad \quad
\label{2_eq.mf}
\left\{
\begin{aligned}
&i\partial_{t}\gamma(t)=A \gamma(t)+ \partial_{\overline{z}}q_{0}(\gamma(t)),&\\
&\gamma(0)=x_0 \in \Z_1\,.&
\end{aligned}
\right.
\end{equation}
This is a semi-linear evolution equation which is usually reinterpreted and studied in the interaction representation.  In fact, if one considers a solution
$\tilde\gamma$ of \eqref{2_eq.mf} then  $\gamma(t)= e^{it A} ~\tilde\gamma(t)$ is formally a solution of the following initial value problem,
\begin{equation}
\label{2_int.IVP}
  \left\{
    \begin{aligned}
    &&\dot{\gamma}(t) \ = v(t,\gamma(t)) \\
    &&\gamma(0) \ = x_{0} \in \Z_{1}, \\
    \end{aligned}
  \right.
\end{equation}
where $v:\R\times \Z_{1}\to \Z_{-1} $ is the non-autonomous vector field given by
\begin{equation}
\label{2_int.v}
 v(t,z):= -ie^{itA} \,\partial_{\overline{z}}q_{0}(e^{-itA}z)\,.
\end{equation}
The following definition introduces the notions of weak solutions for the  initial value problem \eqref{2_int.IVP}  and their  uniqueness. Such definition is motivated by the properties of the vector field $v$ which satisfies in particular the  bound,
    \begin{equation}
    \label{2_int.bndv}
   \exists C > 0 ~,~ \forall (t,z) \in \R \times \Z_{1} ~,~ ||v(t,z)||_{\Z_{- 1}} \leq C\, (||z||_{\Z_{1}}^{2}.||z||_{\Z}^{2}+1) \,,
    \end{equation}
proved in the Appendix B in Lemma \ref{2_Controle de v}.

\begin{dfn}
\label{2_defweaksol}
A weak solution  of the initial value problem \eqref{2_int.IVP}, defined  on  a time interval $0\in I$, is a function $t \in I \longrightarrow \gamma(t)$ belonging to the space $L^{2}(I,\Z_{1}) \cap L^{\infty}(I,\Z) \cap W^{1,\infty}(I,\Z_{-1})$ and  satisfying \eqref{2_int.IVP} for a.e $t \in I$.
In addition, we say that the initial value problem \eqref{2_int.IVP} satisfies the weak uniqueness property  if \eqref{2_int.IVP}  admits at most one weak solution for each initial datum $x_0\in \Z_1$. 
\end{dfn}
Recall that $W^{1,\infty}(I,\Z_{-1})$ here denotes the Sobolev space composed of classes of functions in $L^\infty(I,\Z_{-1})$ with  distributional first derivatives in  $L^\infty(I,\Z_{-1})$. Remember that an element  $\gamma\in W^{1,\infty}(I,\Z_{-1})$ is an absolutely continuous curve in $\Z_{-1}$  with almost everywhere defined derivatives in $\Z_{-1}$ satisfying $\dot\gamma\in L^\infty(I,\Z_{-1})$.
It also makes sense to require an initial condition in \eqref{2_int.IVP},   since the weak solutions are in particular  continues curves valued in $\Z_{-1}$.

A consequence of the above bound \eqref{2_int.bndv} is that it makes sense to consider weak solutions of  the initial value problem \eqref{2_int.IVP} as given by Definition \ref{2_defweaksol} and actually the problem is   equivalent  to the integral equation,
\begin{equation}
\label{2_inteq}
\gamma(t)=x_0+\int_0^t v(s,\gamma(s)) \, ds
\end{equation}
for any $t\in I$.  In fact, if $\gamma$ is a weak solution then the function  $s\in I \to ||v(s,\gamma(s))||_{\Z_{-1}}$ belongs to $L^1(I,ds)$ thanks to the bound \eqref{2_int.bndv} and hence $\gamma$ satisfies  \eqref{2_inteq}. Conversely,  if $\gamma$ is a curve in $L^{2}(I,\Z_{1}) \cap L^{\infty}(I,\Z) \cap W^{1,\infty}(I,\Z_{-1})$ verifying \eqref{2_inteq} then $\gamma$ is an absolutely continuous function satisfying \eqref{2_int.IVP}.

 Remark that the nonlinear Hamiltonian equation \eqref{2_eq.mf} admits the following formal conserved quantities:
\begin{itemize}
\item The charge $||z||_{\Z}$ (following from anti-linear consideration on the vector $ \partial_{\overline{z}} q_{0}(z)$).
\item The classical energy $h(z) = \langle z,Az \rangle + q_{0}(z)= \langle z,Az \rangle + \frac{1}{2}q(z \otimes z, z \otimes z)$.
\end{itemize}
 It is worth noticing that our main result (Theorem \ref{2_main.thm}) does not require existence of solutions nor the conservation of charge and energy for the mean-field equation \eqref{2_int.IVP}. In this respect, it improves significantly  the result of Q.~Liard \cite{MR3661404}. Note also that the nonlinear  equations \eqref{2_eq.mf} and \eqref{2_int.IVP} are  gauge invariant with respect to the unitary group $U(1)$. Such symmetry was recently used in \cite{Ammari:2018aa} to relate the Hartree and Gross–Pitaevskii hierarchies to Liouville's equations.

\bigskip
\textit{Quantum states and Wigner measures:} The mean-field problem is studied here through the Wigner measures method elaborated by Z.~Ammari and F.~Nier in \cite{MR2465733,MR2513969,MR2802894,MR3379490}. The idea of Wigner measures comes from the finite-dimensional semi-classical analysis  and it  has been generalized to infinite dimensional spaces in the latter references. One of the advantages of such method is the possibility to study the mean-field approximation for any initial quantum states without appealing to a coherent structure (like coherent or factorized states) or well-prepared states (like states which are asymptotically factorized).
We recall that a factorized state on $\bigvee^{n} \Z$ is a state of the form $| \phi^{\otimes n} \rangle\langle  \phi^{\otimes n} |$ with $\phi \in \Z$, $||\phi||_{\Z}=1$ and that a coherent state is a state on $\Gamma_s(\Z)$ created from the vacuum by the action of the Weyl operator.

In the sequel, we will consider sequences of normal states on $\bigvee^{N} \Z$ labeled by $N\in\N$ where $N$ represents the number of particles. Remember that a normal state $\rho_{n}$ on $\bigvee^{n} \Z$ is a non-negative  normalized element of the Schatten space of trace class operators  $\mathscr{L}^{1}(\bigvee^{n} \Z)$, i.e.: $\rho_n\geq 0$ and ${\rm Tr}[\rho_{n}]=1$. Before giving  the definition of Wigner measures, we recall briefly the definition of the Weyl operator. For any vector $z \in \Z$, the Weyl operator
$\mathcal{W}(z)$ is the unitary operator,
$$
\mathcal{W}(z)=e^{i \Phi(z)},
$$
 where $\Phi(z) = \frac{1}{\sqrt{2}} (a(z) + a^{*}(z))$  is the field operator and $a(z), a^{*}(z)$ are the annihilation-creation operators
satisfying the canonical commutation relations,
$$
[a(z), a(y)]=0=[a^*(z), a^*(y)]\,, \qquad  [a(z), a^*(y)]=\langle z, y\rangle \,1\,.
$$
For more details about Weyl operators and field operators, we may refer the reader to \cite{MR3060648,MR1441540,MR2465733}.

\begin{dfn}
Let $(\rho_{n})_{n \in \N}$ be a sequence of normal states on $\bigvee^{n} \Z$. The set $\mathcal{M}(\rho_{n} ,~ n \in \N)$ of Wigner measures of $(\rho_{n})_{n \in \N}$ is the set of all Borel probability measures $\mu$ on $\Z$ such that there exists an extraction $\varphi$ satisfying:
\[
\forall \xi \in \Z ~,~ \lim\limits_{ n \to + \infty} \Tr[ \rho_{\varphi(n)} \,\mathcal{W}\big(\sqrt{2\varphi(n)}\pi \xi\big)] = \int_{\Z} e^{2i\pi Re \langle \xi, z \rangle}\, d\mu(z) \,,
\]
where $\mathcal{W}(\sqrt{2\varphi(n)}\pi \xi)$ is the Weyl operator recalled above.
\end{dfn}
It is known that the set of wigner measures is non trivial, see \cite{MR1441540}, and modulo an extraction one can suppose that
 $\mathcal{M}(\rho_{n} ~,~ n \in \N)$  is reduced to a singleton. Usually, for the mean-field problem there is no loss of generality  in assuming a priori  the latter simpler situation. One can compare the set of Wigner measures to the ensemble of  real bounded sequences which generally could have several limit points. In fact, if one wants to prove a property for bounded sequences it is enough sometimes  to prove it only for convergent sequences.

\bigskip
\emph{Main result:}
Consider a sequence  $( \rho_{N})_{N \in \mathbb{N}}$ of normal states on $\bigvee^{N} \Z$. The time evolution of such states through the many-boson dynamics of $H_N$ are given by:
\begin{equation}
\label{2_rhot}
\rho_{N}(t) = e^{-it H_{N}} \rho_{N}  e^{it H_{N}} \,.
\end{equation}
According to the self-adjointness of the operators $H_{N}$, for all $t \in \R$, $\rho_{N}(t)$ is still a normal state. Working in the interaction representation, one considers the time evolved states:
\begin{equation}
\label{2_inter.rho}
\widetilde{ \rho}_{N}(t) = e^{itH_{N}^{0}} \rho_{N}(t) e^{-itH_{N}^{0}} ~.
\end{equation}
Our result on the convergence of the many-boson dynamics towards  the mean-field equation can be summarized as follows.  We prove under some general assumptions that if $ \mathcal{M}(\rho_{N} ,~ N \in \N)=\{\mu_0\}$ at initial time $t=0$ then at any later time the Wigner measures set of  ${ \rho}_{N}(t)$ and $\widetilde{ \rho}_{N}(t)$ are   singletons,
$$
 \mathcal{M}(\rho_{N} ,~ N \in \N)=\{\mu_t\}\, \quad \text{ and } \quad
  \mathcal{M}(\widetilde{ \rho}_{N}(t) ,~ N\in \N)=\{\tilde\mu_t\}\,,
$$
such that  $\mu_t$ and $\tilde\mu_t$  are Bore probability measures  on $\Z$ related for all times according to  the relation,
\begin{equation}
\label{2_mutimu}
\mu_t=(e^{-i t A})_{\sharp} \widetilde \mu_t\,.
\end{equation}
Moreover, there exists a Borel set $\mathcal{G}_0\subset \Z_1$  satisfying $\mu_0(\mathcal{G}_0)=1$ such that for any $x_0\in\mathcal{G}_0$ the initial value problem \eqref{2_int.IVP} admits a unique global weak solution $u(\cdot)$ with the initial condition given by $x_0$ and for any time $t\in \R$  the mapping
\begin{eqnarray*}
\phi(t): \mathcal{G}_0 &\rightarrow & \Z\\
x_0 &\rightarrow & u(t)
\end{eqnarray*}
is well defined, Borel and verifies  the important relation,
$$
\mu_t=\big(e^{-it A}\circ \phi(t)\big)_{\sharp} \mu_0.
$$
Such formulation of the mean-field limit in terms of Wigner measures  is known to imply the convergence of reduced density matrices, see for instance  \cite{MR2802894}.

\medskip
 We now state precisely our main theorem.

\begin{thm}
\label{2_main.thm}
Consider the many-boson Hamiltonian $H_N$ given by \eqref{2_HN} and  assume \eqref{2_A1}-\eqref{2_A2} and one of the assumptions \eqref{2_D1} or \eqref{2_D2}. Furthermore, suppose that the initial value problem \eqref{2_int.IVP} satisfies the  uniqueness property of weak solutions stated in Definition  \ref{2_defweaksol}.
Let $(\rho_{N})_{N \in \mathbb{N}}$ be a sequence of normal states on $\bigvee^{N} \Z$ such that:
\[
 \exists C > 0 ~,~ \forall N \in \mathbb{N} ~,~ \Tr[ \rho_{N} \,H_{N}^{0} ] \leq CN ~,~
\]
and such that the initial set of Wigner measures of $(\rho_{N})_{N \in \mathbb{N}}$ is a singleton, i.e.:
\[
\mathcal{M}( \rho_{N} ~,~ N \in \N) = \lbrace \mu_{0} \rbrace ~.~
\]
Then:
\begin{itemize}
\item For any time $t\in \R$, the set of Wigner measures of the sequence  $(\rho_{N}(t))_{N \in \mathbb{N}}$ given by $ \rho_{N}(t)=e^{-it H_{N}} \rho_{N}  e^{it H_{N}}$  is a singleton, i.e.:
\[
\mathcal{M}( \rho_{N} (t)~,~ N \in \N) = \lbrace \mu_{t} \rbrace ~.~
\]
\item  There exists a Borel set $\mathcal{G}_0\subset \Z_1$  satisfying $\mu_0(\mathcal{G}_0)=1$ such that for any initial condition $x_0\in\mathcal{G}_0$ the initial value problem \eqref{2_int.IVP} admits a unique global weak solution $u(\cdot)$  and for any time $t\in \R$  the mapping
\begin{eqnarray*}
\phi(t): \mathcal{G}_0 &\rightarrow & \Z\\
x_0 &\rightarrow & u(t)
\end{eqnarray*}
is well defined and Borel.
\item The Wigner measure $\mu_t$ is  identified for all times  as,
$$
\mu_t=\big(e^{-it A}\circ \phi(t)\big)_{\sharp} \mu_0.
$$
\end{itemize}
\end{thm}

\begin{remark}
A consequence of the main Theorem \ref{2_main.thm} is the existence of a global generalized  flow for the initial value problems \eqref{2_int.IVP} and \eqref{2_eq.mf}.    This means that the assumptions \eqref{2_A1}-\eqref{2_A2} and \eqref{2_D1} or \eqref{2_D2}, imply the existence of global weak solutions for the initial value problem \eqref{2_int.IVP}.  Such observation seems  new and its proof is quite different from the similar  classical  result  in \cite[Theorem 3.3.9]{MR2002047}.	

\end{remark}

\begin{sproof}
Here we briefly sketch the key points in the  proof of the main theorem.
\begin{itemize}
\item[(i)] Consider an initial  sequence of normal states $(\rho_{N})_{N \in \N}$. It is interesting to  work with the
interaction representation. Indeed, the well-defined quantum dynamics leads to a sequence of evolved states $(\widetilde{\rho}_{N}(t))_{N \in \N} ~,~ t \in \R$, given by \eqref{2_inter.rho}. Moreover, one easily checks that
the set of Wigner measures of $(\rho_{N}(t))_{N \in \N}$ and $(\widetilde{\rho}_{N}(t))_{N \in \N} $ are related
according to the relation \eqref{2_mutimu}. So, it is enough to focus only in  the latter sequence.
A priori the set of Wigner measures of  $\widetilde{\rho}_{N}(t)$ may not be a singleton and their measures are not a priori  known or identified.
\item[(ii)] Using a standard extraction argument recalled in Proposition \ref{2_mu boule}, one can prove that for any subsequence of  the  states $({\rho}_{N})_{N \in \N}$, there exists an extraction $\psi$  such that for all times
$\mathcal{M}( \widetilde{\rho}_{\psi(N)}(t) ~,~ N \in \N) = \lbrace \widetilde\mu_{t} \rbrace $.
\item[(iii)]   Applying the convergence argument of \cite{MR3661404} recalled in Proposition \ref{2_rho charac}, one shows that the obtained  curve of probability measures   $t \in \R \to \widetilde\mu_{t}\in\mathfrak{P}(\Z)$ satisfies a fortiori the  characteristic equation \eqref{2_eqchar}.
\item[(iv)] Such  characteristic equation \eqref{2_eqchar} is proved to be equivalent to  a Liouville equation \eqref{2_Liou} in Section \ref{2_liouv.sec}.
\item[(v)] Then  in Section \ref{2_Probrep.sec}  a result is proved on the uniqueness of solutions  for a Liouville equation and a generalized flow for the initial value problem \eqref{2_int.IVP}  is constructed.
\item[(vi)] So, any Wigner measures $\widetilde\mu_t$ of the states $(\widetilde\rho_{N}(t))_{N \in \N}$ should be the unique solution, at time $t$, of the Liouville equation with the initial condition $\mu_0$. Moreover, $\widetilde\mu_t$ is identified as the push-forward  of  $\mu_0$ by the generalized flow of the initial value problem \eqref{2_int.IVP}.
\item[(vii)] The simple relations between $(\rho_{N}(t))_{N \in \N}$, $(\widetilde{\rho}_{N}(t))_{N \in \N}$ and
$\mu_t$, $\widetilde\mu_t$   yield the final result.
\end{itemize}
\end{sproof}

\bigskip
\emph{Overview of the article:}
The paper is organized as follows.  Several  examples satisfying the statement of  Theorem
\ref{2_main.thm} are listed in Section \ref{2_sec.examples}.  In the next Section \ref{2_liouv.sec},  we introduce  the Liouville equation and we prove its  equivalence to a characteristic evolution  equations. Such result holds true for general initial value problems and the related vector field need not be the one of the mean-field equation.  In Section \ref{2_conv.sec},  we consider the set of Wigner measures of time evolved  quantum states and recall the convergence results of Q.~Liard \cite{MR3661404}. Moreover, we prove that the Fourier transform of these Wigner probability measures satisfy the characteristic evolution  equation of Section \ref{2_liouv.sec} with the vector field given by \eqref{2_int.v}.  Hence, one  concludes that the Wigner measures at hand  verify a Liouville's equation.  Finally, we prove  in Section \ref{2_Probrep.sec} a uniqueness result for Liouville's equations and construct a generalized flow
for the initial value problem  \eqref{2_int.IVP}.   This section relies on probabilistic representation ideas  related to the kinetic theory and the optimal transport theory \cite{MR2129498, MR2400257}. Surprisingly, such techniques  turn to be remarkably efficient in solving the mean field  problem. Indeed, they allow to prove elegantly  the  uniqueness of solutions of Liouville's equations in infinite dimensional spaces, see \cite{Ammari:2018aa,MR3721874,MR3379490}.  For the reader's convenience, some useful results are collected  in  the Appendix \ref{2_appA} concerning the  measurability of certain maps and sets.  In Appendix  \ref{2_appB}, we prove a  useful bound satisfied by the
vector field \eqref{2_int.v} related to the mean-field equation \eqref{2_int.IVP}.

\section{Examples}
\label{2_sec.examples}
Several relevant examples are provided in this section for the illustration of our main Theorem \ref{2_main.thm}.  They concern
the LLL model,  the compound Bose gases model, the non-relativistic and the semi-relativistic models of  quantum mechanics.
Before listing these examples, we give a simple argument that ensures the uniqueness proprety  for the initial value problem \eqref{2_int.IVP} with the vector field $v$ given by \eqref{2_int.v}.  In particular, the proposition  below shows that our main Theorem \ref{2_main.thm} applies to  all the examples of \cite{MR3661404} where the condition  \eqref{2_est}  is verified.

\medskip
Recall  that $q_0$ is the quadratic form defined in \eqref{2_intr.q0}.
\begin{prop}
Suppose that for any $M > 0$ there exists $C(M) > 0$ such that:
\begin{equation}
\label{2_est}
||\partial_{\overline{z}}q_{0}(x)-\partial_{\overline{z}}q_{0}(y)||_{\Z} \leq C(M) (||x||_{\Z_{1}}^{2}+||y||_{\Z_{1}}^{2})||x-y||_{\Z} ~,~
\end{equation}
for all $(x,y) \in \Z_{1}^{2}$ such that $(x,y) \in B_{\Z}(0,1)$. Then the initial value problem
\eqref{2_int.IVP} with the vector field $v$ in \eqref{2_int.v} satisfies the uniqueness property of Definition \ref{2_defweaksol}.
\end{prop}

\begin{proof}

Let $\gamma_{1} \in L^{2}({I},\Z_{1}) \cap L^{\infty}({I},\Z)\cap W^{1,\infty}(\overline{I},\Z_{- 1})$ and $\gamma_{2} \in L^{2}(\overline{I},\Z_{1}) \cap L^{\infty}(\overline{I},\Z) \cap W^{1,\infty}(\overline{I},\Z_{- 1})$ be two weak solutions of \eqref{2_int.IVP} with the same initial condition at time $t_{0}=0 \in I$. Then, the following equality holds true in $C(\overline{I},\Z_{- 1})$:

\[ \gamma_{1}(t)-\gamma_{2}(t) = \int_{t_{0}}^{t} v(s,\gamma_{1}(s))-v(s,\gamma_{2}(s))\,ds \]
\medskip
And so, as $\gamma_{1}$ and $\gamma_{2}$ are in $\Z$ for all $t \in \overline{I}$ (see remark \eqref{2_rmkcontinuitycurve}) we have:
\begin{eqnarray*}
|| \gamma_{1}(t)-\gamma_{2}(t)||_{\Z} &\leq& \int_{t_{0}}^{t} ||v(s,\gamma_{1}(s))-v(s,\gamma_{2}(s))||_{\Z}
\,ds \\
&\leq& \int_{t_{0}}^{t} ||\partial_{\overline{z}}q_{0}(\gamma_{1}(s))-\partial_{\overline{z}}q_{0}(\gamma_{2}(s))||_{\Z} \,ds\,.
\end{eqnarray*}
We now have to remind that $\gamma_{1}(t)$ and $\gamma_{2}(t)$ belong to the closed ball ${B}_{\Z}(0,1)$ for all $t \in \overline{I}$ and so the assumption \eqref{2_est} gives the existence of a constant $C(1) > 0$ such that:

\[ || \gamma_{1}(t)-\gamma_{2}(t)||_{\Z} \leq \int_{t_{0}}^{t} C(1)(||\gamma_{1}(s)||_{\Z_{1}}^{2}+||\gamma_{2}^{s}||_{\Z_{1}}^{2})||\gamma_{1}(s)-\gamma_{2}(s)||_{\Z} \,ds ~,~
\]
and if we denote $g(s) = ||\gamma_{1}(s)||_{\Z_{1}}^{2}+||\gamma_{2}^{s}||_{\Z_{1}}^{2}$ then $g \in L^{1}(\overline{I},\Z_{1})$ and then

\[ || \gamma_{1}(t)-\gamma_{2}(t)||_{\Z} \leq C(1)  \int_{t_{0}}^{t}  g(s) ||\gamma_{1}(s)-\gamma_{2}(s)||_{\Z} \,ds\,.
\]
The Gronwall's lemma shows  that $\gamma_{1}=\gamma_{2}$ on $\overline{I}$. Hence the weak uniqueness property of weak solutions of \eqref{2_int.IVP}  is satisfied according to the  Definition \ref{2_defweaksol}.
\end{proof}

Our main Theorem \ref{2_main.thm} applies to all the examples given below. The main assumptions to be checked are \eqref{2_A1}-\eqref{2_A2} and \eqref{2_D1} or \eqref{2_D2}. They are actually  not difficult to prove and we refer the reader to \cite{MR3661404} for more details. Actually, the only ingredients we need to specify are the Hilbert space $\Z$, the one-particle operator  $A$ and the two-particle interaction $q$.

\bigskip
\noindent
\emph{The LLL model:}
Such model appears for instance in \cite{MR2570761} where a system of  trapped bosons undergoing rapid rotation is studied and an effective Hamiltonian  with respect to the Lowest Landau Level is derived. The latter Hamiltonian describes a many-boson system with
\begin{equation}
\Z=\mathcal{H}L^2(\C, e^{-|z|^2}  L(dz))\,, \qquad
A=z\cdot \partial_z\,,\qquad
q=\delta_{12}\,,
\end{equation}
where $\Z$ is the Segal-Bargman space of holomorphic functions that are  square integrable with respect to the gaussian measure $\gamma=e^{-|z|^2}  L(dz)$ ( with $ L(dz)$ is the Lebesgue measure on $\C$), $A$ is the harmonic oscillator  in this Fock or Segal-Bargman  representation (see e.g.~\cite{MR983366}) and  $q$ is a bounded sesquilinear form  given explicitly by
\begin{equation*}
\langle F, \delta_{12} \,G\rangle_{\Z^{\otimes 2}}=
 \int_{\C} \overline{F(\xi,\xi)} \,
G(\xi,\xi) e^{-2|\xi|^2}  \, L(d\xi)  \,,
\end{equation*}
for any $F,G\in \Z^{\otimes 2}\equiv \mathcal{H}L^2 (\C^2, e^{-|\xi|^2} L(d\xi))$.  Hence, the assumptions \eqref{2_A1}, \eqref{2_A2}, \eqref{2_D1} and \eqref{2_est}  are trivially satisfied  since 
$q$ is a bounded sesquilinear form.

\bigskip
\noindent
\emph{The compound Bose gases  model:}
This second example is inspired by the recent works in two compound interacting Bose gases \cite{MR3711618,MR3681700}. The many-body Hamiltonian is specified by the following choice:
\begin{equation*}
\Z=L^2(\R^d,\C)\oplus L^2(\R^d,\C)\,,\qquad A=S\oplus S\,,
\end{equation*}
where $S$ is the Laplace operator in $d$-dimension,
$$
S=-\Delta_x\,.
$$
In order to introduce a two particle interaction $q$, note that the space $\Z^{\otimes 2}$ can be identified with $L^2(\R^{2d},\C)\otimes \C^2\otimes \C^2$ and that any vector $u\in\Z^{\otimes 2}$ admits the  decomposition,
$$
u=\sum_{i,j=1,2} u^{(i,j)}(x,y) \,e_i\otimes e_j\,,
$$
where $\{e_1,e_2\}$ is the canonical O.N.B of $\C^2$ and $u^{(i,j)}\in L^2(\R^{2d},\C)$.
Moreover, the sesquilinear form $q$ is given precisely as,
\begin{eqnarray*}
q(u,u)=\sum_{i,j=1,2} \langle u^{(i,j)}, V^{(i,j)}(x-y) \,u^{(i,j)}\rangle_{L^2(\R^{2d},\C)}\,,
\end{eqnarray*}
 where $V^{(i,j)}$ are real-valued functions satisfying
 $$
 V^{(i,j)}\in L^{\alpha}(\R^{d})+L_{0}^{\infty}(\R^{d})\,, \qquad \text{ with } \; \alpha\geq \max(1,\frac{d}{2})
 \quad ( \alpha>1 \text{ if } d=2).
$$
  Here $L_{0}^{\infty}(\R^{d})$ denotes the space of bounded measurable functions converging to $0$ at infinity.
  This class includes in particular  the Yukawa and Coulomb potentials $V^{(i,j)}(x)=g\frac{e^{-\lambda |x|}}{|x|^{\beta}}$, $\lambda\geq 0$, $g\in\R$, $0<\beta <2$ and $d=3$. Again, it is easy to check the assumptions \eqref{2_A1}, \eqref{2_A2}   and the estimate \eqref{2_est} using Young, H\"older and Sobolev inequalities. Moreover, the assumption  \eqref{2_D2} follows from the compactness of the operator
$ (1+S)^{-1/2} V^{(i,j)}(x) (1+S)^{-1/2}$, see e.g.~\cite[Lemma 3.10]{MR3379490}.

\bigskip
\noindent
\emph{The non-relativistic models:}  Non-relativistic models of quantum mechanics are widely used and studied. Here we give three examples: The first is a delta interaction model  in one dimension appearing in nuclear physics and was studied for instance in \cite{MR2953701,MR2331036}. The second is a confined many-boson system while the third is the typical  translation invariant model of quantum mechanics.

\begin{itemize}
\item  \emph{Delta interaction:}
\begin{equation*}
\Z=L^2(\R, dx)\,,\quad  A=-\Delta_x+V(x)\,,\quad q=\lambda\, \delta(x-y)\,,
\end{equation*}
where $\lambda\in\R$, $V\in L^1_{loc}(\R)$, $V\geq 0$ and $\lim_{|x|\to\infty}V(x)=+\infty$ is a confining potential and $q$ is the delta potential defined explicitly for sufficiently regular $u\in \Z^{\otimes 2}=L^2(\R^2)$  as
$$
q(u,w)=\int_{\R} \overline{u(x,x)} w(x,x) \,dx\, .
$$
\item  \emph{Confined many-boson system:}
\begin{equation*}
\Z=L^2(\R^d, dx)\,,\quad  A=-\Delta_x+V(x)\,,\quad q= W(x-y)\,,
\end{equation*}
where $V\in L^1_{loc}(\R^d)$, $V\geq 0$ and $\lim_{|x|\to\infty}V(x)=+\infty$ is a confining potential and
$W\in L^\alpha(\R^d)+L^\infty(\R^d)$ such that $\alpha\geq \max(1,\frac d 2)$ and ($\alpha>1$ if $d=2$).
\item  \emph{Translation invariant system:}
 \begin{equation*}
\Z=L^2(\R^d, dx)\,,\quad  A=-\Delta_x\,,\quad q= W(x-y)\,,
\end{equation*}
where $W\in L^\alpha(\R^d)+L_0^\infty(\R^d)$ such that $\alpha\geq \max(1,\frac d 2)$ and ($\alpha>1$ if $d=2$).
\end{itemize}
In the first two cases the assumptions \eqref{2_A1}, \eqref{2_A2},  \eqref{2_est}  and   \eqref{2_D1}
are satisfied using Young, H\"older and Sobolev inequalities. In the last case one verifies \eqref{2_D2}  thanks to the compactness of $ (1-\Delta_x)^{-1/2} W(x) (1-\Delta_x)^{-1/2}$.


\bigskip
\noindent
\emph{The semi-relativistic model:}
The semi-relativistic  many-boson Hamiltonian is described by the following choice:
\begin{equation*}
\Z=L^2(\R^3,dx)\,, \qquad A=\sqrt{-\Delta_{x}+m^2}+V(x)\,, \qquad q=\frac{\kappa}{|x-y|} \,,
\end{equation*}
where $|\kappa|<\kappa_{cr}$ with  $\kappa_{cr}^{-1}:=\displaystyle 2\lim_{\alpha\to+\infty}  ||\frac{1}{|x|} (-\Delta+\alpha)^{-\frac{1}{2}}||$, $m\geq 0$ and $V$ is real-valued measurable function  satisfying $V=V_1+V_2$ with:
\begin{eqnarray*}
&&V_1\in L^1_{loc}(\R^3), \, V_1\geq 0,\, \lim_{|x|\to \infty} V_1(x)=+\infty\,,\\
&&V_2 \text{ is form bounded by } \sqrt{-\Delta} \text{ with a relative bound less than } 1\,.
\end{eqnarray*}
The assumptions \eqref{2_A1},\eqref{2_A2}, \eqref{2_est} and \eqref{2_D1} are checked
in \cite[Example 5]{MR3661404}.

\section{Liouville equation}
\label{2_liouv.sec}
The Wigner measures method for the mean field problem is based on the analysis of a Liouville equation (or a continuity equation) which is  derived naturally from the dynamics of the quantum many-body system.  In this section, we briefly recall the notion of Liouville equation related to a general initial value problem and we prove its equivalence with a characteristic evolution equation. Due to the independent interest that presents such result, we keep this part as general as possible. So, in this section we do not assume that the initial value problem \eqref{2_int.IVP}  is specifically the one given by the vector field in \eqref{2_int.v}. Instead, we  consider the vector field $v : \mathbb{R} \times \Z_{1} \to \Z_{-1}$ as a Borel map which is bounded on bounded sets.

The Liouville equation associated to the initial value problem \eqref{2_int.IVP} is formally  given  by the equation:
\[
\partial_{t} \mu_{t} + \nabla^{T} (v(t,.) \cdot\mu_{t}) = 0\,.
 \]
To give a rigorous  meaning to the above identity we use smooth cylindrical functions as tests functions and interpret the equation in a distributional  sense. Below, we recall briefly the sense given to this equation. To see further details about cylindrical functions and Liouville equation, see for instance \cite{MR3721874,MR2129498,MR3379490}.

\bigskip

Consider $\Z_{- 1}$ as a real Hilbert space denoted by $\Z_{- 1, \mathbb{R}}$ and endowed with the scalar product $Re \langle .,. \rangle_{\Z_{-1}}$. For $n \in \N$, we define $\mathbb{P}_{n}$ as the set of all projections $\pi : \Z_{- 1, \mathbb{R}} \to \mathbb{R}^{n}$ such that:
\[
\pi:x \to \pi(x) = ( \langle x,e_{1} \rangle_{\Z_{- 1, \mathbb{R}}},..., \langle x,e_{n} \rangle_{\Z_{- 1, \mathbb{R}}}) \,,
\]
where $(e_{1},...,e_{n})$ is an orthonormal family of $\Z_{- 1, \mathbb{R}}$.
The set of cylindrical functions on $\Z_{- 1}$, denoted $ \mathscr{C}_{0,cyl}^{\infty}( \Z_{- 1})$, is by definition the space of functions $\phi = \psi \circ \pi$ where $\pi \in \mathbb{P}_{n}$ for some $n \in \mathbb{N}$ and $ \psi \in \mathscr{C}_{0}^{\infty}(\mathbb{R}^{n})$.
Furthermore, for  a given open bounded interval $I$, the set of cylindrical functions on $I \times \Z_{- 1}$, denoted $ \mathscr{C}_{0,cyl}^{\infty}(I \times \Z_{- 1})$, is defined as the set of functions $\phi$ such that :
\[
 \forall (t,x) \in I \times \Z_{- 1} ~,~ \phi(t,x) = \phi(t,\pi(x)) ~,~
 \]
where $\pi \in \mathbb{P}_{n}$ for some $n \in \mathbb{N}$ and $\phi \in \mathscr{C}_{0}^{\infty}(I \times \mathbb{R}^{n})$.

\bigskip

Let $\lbrace \mu_{t} \rbrace_{ t \in I}$ be a family of Borel probability measures on $\Z_{- 1}$. We say that  $\mu_{t}$ satisfies the Liouville equation on $I$, associated to the vector field $v$, if:
\begin{equation}
\label{2_Liou}
     \forall \phi \in \mathscr{C}_{0,cyl}^{\infty}(I \times \Z_{- 1}) ~,~ \int_{I} \int_{\Z_{- 1}} \partial_{t} \phi(t,x) + {Re} \langle v(t,x), \nabla_{\Z_{-1}, \mathbb{R}} ~\phi(t,x) \rangle_{\Z_{- 1}} ~d\mu_{t}(x)dt=0 \,,
\end{equation}
where  $\nabla_{\Z_{-1}, \mathbb{R}}$ denotes the $\mathbb{R}$-gradient in $\Z_{- 1,\mathbb{R}}$ and $ \mu_{t}$ is supposed also to be a Borel probability measures on  $\Z_{1}$ for a.e.~ $t\in I$.  Of course up to now the equation  \eqref{2_Liou} may not make sense and  some further requirements are need to make it consistent. In particular,  a reasonable bound on the vector field $v$ and a certain continuity with respect to time of the curve  $t\in I\to \mu_{t}$   are  more or less  necessary.  For this, we introduce the following useful topologies in the set of Borel probability measures  $\mathfrak{P}(\Z_{-1})$.

\begin{dfn}
\label{2_defnarowcv}
Let $ (\mu_{n})_{n \in \N}$ be a sequence of measures in $\mathfrak{P}(\Z_{-1})$ and let $\mu \in \mathfrak{P}(\Z_{-1})$.
\begin{itemize}
\item We say that $ (\mu_{n})_{n \in \N}$ converges strongly narrowly to $\mu$ and we denote it $ \mu_{n} \to \mu$ if:
\[
\forall f \in \mathscr{C}_{b}(\Z_{-1}) ~,\quad \int_{\Z_{-1}} f d\mu_{n} \longrightarrow  \int_{\Z_{-1}} f d\mu \,,
\]
where  $ \mathscr{C}_{b}(\Z_{-1})$ denotes the space of bounded continuous functions on $\Z_{-1}$ endowed with it norm topology.
\item We say that $ (\mu_{n})_{n \in \N}$ converges weakly narrowly to $\mu$ and we denote it $\mu_{n} \rightharpoonup \mu$ if:
\[
 \forall f \in \mathscr{C}_{b}(\Z_{-1,w}) ~,\quad \int_{\Z_{-1}} f d\mu_{n} \longrightarrow  \int_{\Z_{-1}} f d\mu\,,
 \]
where  $ \mathscr{C}_{b}(\Z_{-1,w})$ denotes the space of bounded continuous functions on $\Z_{-1}$ endowed with the weak norm
$$
||u||_{w}^2=\sum_{k=0}^\infty \frac{1}{2^k} |\langle u, e_k\rangle_{\Z_{-1}}|^2\,
$$
where $\{e_k\}_{k\in\N}$ is an orthonormal basis of $\Z_{-1}$.
\end{itemize}
\end{dfn}
Regarding the strong (resp.~weak) topology on $\mathfrak{P}(\Z_{-1})$, one defines a family of measures $\lbrace \mu_{t} \rbrace_{t \in I}$ to be strongly (resp.~weakly) narrowly continuous if, for any $t_{0} \in I$, for any $(t_{n})_{n \in \mathbb{N}} \in I^{\N}$ such that $t_{n} \to t_{0}$ then $\mu_{t_{n}}$ converges  strongly (resp.~weakly) narrowly  to $\mu_{t_{0}}$.

 In the sequel, the set of measures $\lbrace \mu_{t} \rbrace_{t \in \R}$ that we shall  consider arises as the Wigner measures of time-dependent  quantum states. Our aim is to show that these measures satisfy the Liouville eqaution with some reasonable assumptions. Commonly, and according to \cite{Ammari:2018aa}, the theorems we will use later on requests the following conditions:
\begin{itemize}
\item \textit{concentration:}  For a.e.~ $t \in I ~,~ \mu_{t} \in \mathfrak{P}(\Z_{1})$.
\item \textit{norm control of the vector field relatively to the measures:}
\begin{equation}
\label{2_condliou}
\int_{I} \int_{\Z_{1}} || v(t,x) ||_{\Z_{- 1}} d\mu_{t}(x)dt < + \infty ~.
 \end{equation}
\item \textit{regularity of the measures:} $t \in I \to \mu_{t}$ is weakly narrowly continuous in $\mathfrak{P}(\Z_{- 1})	$.
\end{itemize}

The following proposition shows the equivalence between two mathematical formulations of a transport phenomena, namely a Liouville or continuity  equation and  a characteristic equation.  While the first deals with probability measures the second considers
their characteristic functions using Fourier transform.  This result  will be useful latter on to related the quantum dynamics of many-boson systems to a Liouville equation in the limit of a large number of particles.

\begin{prop}\label{2_Liouville and characteristic equation}

Let $I$ be an open bounded time interval and $ v : \R \times \Z_{1} \to \Z_{- 1}$ be a vector field such that  $v$ is Borel and bounded on bounded sets. Consider $t\in I\to\mu_{t}$ a curve in $\mathfrak{P}(\Z_{- 1})$ satisfying:
\begin{itemize}
\item  $t \in I \to \mu_{t}$ is weakly narrowly continuous in $\mathfrak{P}(\Z_{- 1})	$.
\item For a.e.~$t \in I$, $\mu_{t} \in \mathfrak{P}(\Z_{1})$.
\item The following bound holds true,
$$ \int_{I} \int_{\Z_{-1}} || v(t,x) ||_{\Z_{- 1}} \,d\mu_{t}(x)dt < + \infty ~. $$
\end{itemize}
Then the following assertions are equivalent :
\begin{itemize}
\item[(i)]  $ \lbrace \mu_{t} \rbrace_{t \in I}$ is a solution of the Liouville equation \eqref{2_Liou}.
\item[(ii)]  $ \lbrace \mu_{t} \rbrace_{t \in I}$ solves the following characteristic equation:
$\forall (t,t_{0}) \in I^{2}$, $ \forall y \in \Z_{1}$,
\begin{equation}
\label{2_eqchar2}
    \mu_{t}(e^{2i\pi Re\langle y,.\rangle_{\Z}}) = \mu_{t_{0}}(e^{2i\pi Re\langle y,.\rangle_{\Z}}) +
     2i\pi \int_{t_{0}}^{t} \mu_{s} \big(e^{2i\pi Re \langle y,x \rangle_{\Z}} Re\langle v(s,x),y \rangle_{\Z}\big) \,ds ~,~
\end{equation}
where we have used the notation $\mu_{t}(e^{2i\pi Re \langle y,.\rangle_{\Z}}) = \int_{\Z_{1}}e^{2i\pi Re \langle y,x \rangle_{\Z}} d\mu_{t}(x)  $.
\end{itemize}
\end{prop}

\begin{proof}

This proof is similar to  the one in \cite[Proposition 3.7]{Ammari:2018aa}. The main difference concerns the hypothesis made on the curve $\lbrace \mu_{t} \rbrace_{t \in I}$ which are less restrictive in our case (we do not require $\mu_{t}(B_{\Z_{1}}(0,R))=1$ for some $R>0$). For the sake of clarity  we give a detailed proof.

Suppose that $ \lbrace \mu_{t} \rbrace_{t \in I}$ is a solution of the Liouville equation \eqref{2_Liou} satisfying the assumptions of the proposition. First note that The Liouville equation makes sense under the
assumptions stated above. Consider now a test function $\phi(t,x) = \chi(t)\phi_{m}(x)$ with $\chi \in \mathscr{C}_{0}^{\infty}(I)$ and $\phi_{m}$ given by:
\[
\phi_{m}(x)= cos(2\pi Re \langle z,x \rangle_{\Z_{- 1}}) \psi(\frac{Re \langle z,x \rangle_{\Z_{- 1}}}{m}) \,,
\]
for a fixed $z \in \Z_{- 1}$ and $\psi \in  \mathscr{C}_{0}^{\infty}(\R)$, $ 0 \leq \psi \leq 1$ and $\psi$ equal to $1$ in a neighborhood of $0$. With this choice, the functions $\phi_{m}$ pointwisely converge to $cos(2\pi Re \langle z,. \rangle_{\Z_{- 1}})$ as $m\to\infty$. One can easily check that for every $m \in \N$, $\phi$ is a smooth cylindrical function in $ \mathscr{C}_{0,cyl}^{\infty}(I \times \Z_{- 1})$. Since the curve $ \lbrace \mu_{t} \rbrace_{t \in I}$ is a solution of the Liouville equation \eqref{2_Liou}, then
\[
\int_{I} \int_{\Z_{1}} \chi^{'}(t)\phi_{m}(x) + Re \langle v(t,x),\nabla \phi_{m}(x) \rangle_{\Z_{- 1}} \chi(t) \,d\mu_{t}(x)dt = 0
\,.
\]
Remark that the above integral is well-defined since
\[
\int_{I} \int_{\Z_{1}} || v(t,x) ||_{\Z_{- 1}}\, d\mu_{t}(x)dt < + \infty ~.~
\]
The gradient of $\phi_{m}$ can be calculated easily and one has:
\[
\nabla \phi_{m}(x) = -2 \pi sin(2\pi Re \langle z,x \rangle_{\Z_{- 1}})\,\psi(\frac{Re \langle z,x \rangle_{\Z_{- 1}}}{m}) \cdot z +
\]
\[ cos(2\pi Re \langle z,x \rangle_{\Z_{- 1}})\frac{1}{m} \psi^{'}(\frac{Re \langle z,x \rangle_{\Z_{- 1}}}{m})\cdot z \in \Z_{- 1}
\]
Replacing this expression in the previous distributional equation and taking the limit $m\to\infty$ then the dominated convergence theorem leads to:
\[
\int_{I} \chi^{'}(t) \int_{\Z_{1}} cos(2\pi Re \langle z,x \rangle_{\Z_{- 1}}) d\mu_{t}(x)dt = \int_{I} \chi(t) \int_{\Z_{1}} 2 \pi sin(2\pi Re \langle z,x \rangle_{\Z_{- 1}}) Re \langle v(t,x),x \rangle_{\Z_{- 1}} d\mu_{t}dt.
\]
Using the same argument with $\phi_{m}(x) =  sin(2\pi Re \langle z,x \rangle_{\Z_{- 1}}) \psi(\frac{Re \langle z,x \rangle_{\Z_{- 1}}}{m})$ one obtains a similar integral equation with sinus. A sum of these two equations leads to the integral equation:
\[
\int_{I} \chi^{'}(t) \int_{\Z_{1}} e^{2i\pi Re \langle z,x \rangle_{\Z_{- 1}}} d\mu_{t}(x)dt = -2i\pi
\int_{I} \chi(t) \int_{\Z_{1}} Re \langle v(t,x),z \rangle_{\Z_{- 1}} e^{2i\pi Re \langle z,x \rangle_{\Z_{-1}}} d\mu_{t}(x)dt\,.
\]
If we set
\[
m(t):= \int_{\Z_{1}} e^{2i\pi Re \langle z,x \rangle_{\Z_{- 1}}} d\mu_{t}(x) ~,~
\]
Then the previous equation becomes in a distributional sense,
\[
\frac{d}{dt} m(t) = 2i\pi  \int_{\Z_{1}} Re \langle v(t,x),z \rangle_{\Z_{- 1}} e^{2i\pi Re \langle z,x \rangle_{\Z_{-1}}} d\mu_{t}(x) \,.
\]
So that $m \in W^{1,1}(I,\mathbb{C})$ and $m$ is an absolutely continuous function over $I$. Hence,  the fundamental theorem of analysis holds true:
\[
\forall (t,t_{0}) \in I^{2} ~,~ m(t) = m(t_{0}) + \int_{t_{0}}^{t} m^{'}(s) ds \,.
\]
This equation is nothing less than the characteristic equation \eqref{2_eqchar2} if we set $y=A^{- 1} z \in \Z_{1}$:
\[
\int_{\Z_{1}} e^{2i\pi Re \langle y,x \rangle _{\Z}} d\mu_{t}(x) = \int_{\Z_{1}} e^{2i\pi Re \langle y,x \rangle_{\Z}} d\mu_{t_{0}}(x) +  2i\pi \int_{t_{0}}^{t} \int_{\Z_{1}} Re \langle v(t,x),y \rangle_{\Z} e^{2i\pi Re \langle y,x \rangle_{\Z}} d\mu_{s}(x) ds .
  \]

\bigskip
Conversely, suppose that the set of measures $ \lbrace \mu_{t} \rbrace_{t \in I}$ solves the characteristic equation \eqref{2_eqchar2}. Let $\psi \in \mathscr{C}_{0,cyl}^{\infty}(\Z_{- 1})$. Then there exists an orthogonal projection $p$ on a finite dimensional subspace of $\Z_{- 1}$ such that $\psi(x) = \phi(p(x))$ where $\phi \in \mathscr{C}_{0}^{\infty}(p(\Z_{- 1}))$. So, the function $\psi$ is sufficiently  smooth and  one can use inverse Fourier transform to write:
\[
\psi(x) = \int_{p(\Z_{- 1})} cos(2 \pi Re \langle z,x \rangle_{\Z_{- 1}}) \mathcal{F}_{R}(\psi)(z) + sin(2 \pi Re \langle z,x \rangle_{\Z_{- 1}}) \mathcal{F}_{I}(\psi)(z) dL(z) ~,~
\]
where $dL$ denotes the Lebsegue measure on $p(\Z_{- 1})$ and
\[
\mathcal{F}_{R}(\psi)(z) =  \int_{p(\Z_{- 1})} cos(2 \pi Re \langle z,x \rangle_{\Z_{- 1}}) \psi(x) dL(x) ~,~
\]
\[
\mathcal{F}_{I}(\psi)(z) =  \int_{p(\Z_{- 1})} sin(2 \pi Re \langle z,x \rangle_{\Z_{- 1}}) \psi(x) dL(x) ~.~
 \]
Splitting the characteristic equation \eqref{2_eqchar2} into real and imaginary part yields:
\[
A = \int_{\Z_{- 1}}  cos(2 \pi Re \langle z,x \rangle_{\Z_{- 1}}) d\mu_{t}(x) = \int_{\Z_{- 1}}  cos(2 \pi Re \langle z,x \rangle_{\Z_{- 1}}) d\mu_{t_{0}}(x) +
\]
\[
\int_{t_{0}}^{t} \int_{\Z_{- 1}} Re(\langle v(s,x),z \rangle_{\Z_{- 1}})(-2\pi sin(2\pi Re \langle z,x \rangle_{\Z_{- 1}})) d\mu_{s}(x)ds   \,,
\]
\[ B = \int_{\Z_{- 1}}  sin(2\pi Re \langle z,x \rangle_{\Z_{- 1}}) d\mu_{t}(x) = \int_{\Z_{- 1}}  sin(2\pi Re \langle z,x \rangle_{\Z_{- 1}}) d\mu_{t_{0}}(x) + \]
\[
\int_{t_{0}}^{t} \int_{\Z_{- 1}} Re(\langle v(s,x),z \rangle_{\Z_{- 1}})(2\pi cos(2\pi Re \langle z,x \rangle_{\Z_{- 1}})) d\mu_{s}(x)ds \,.
\]
Now, a simple computation of $\int_{p(\Z_{- 1})}   (\mathcal{F}_{R}(\psi) \times A +  \mathcal{F}_{I}(\psi) \times B ) dL(z) $ yields
\[
\int_{\Z_{- 1}} \psi(x) d\mu_{t}(x) = \int_{\Z_{- 1}} \psi(x) d\mu_{t_{0}}(x) + \int_{t_{0}}^{t} \int_{\Z_{- 1}} Re( \langle v(s,x),\nabla \psi(x) \rangle_{\Z_{- 1}}) d\mu_{s}(x)ds ~.~
 \]
This equality shows in the distributional sense that:
\[
 \frac{d}{dt} \int_{\Z_{- 1}} \psi(x) d\mu_{t}(x) = \int_{\Z_{- 1}} Re \langle v(t,x),\nabla \psi(x)
 \rangle_{\Z_{-1}} d\mu_{t}(x) \,.
 \]
Multiplying the last  equality by $\chi(t)$ for $\chi \in \mathscr{C}_{0}^{\infty}(I)$ and integrating with respect to time, one obtains that $t\in I\to \mu_{t}$ satisfies the Liouville equation \eqref{2_Liou} for any tests functions of the form $\phi(t,x) = \chi(t)\psi(x)$ with $\psi \in \mathscr{C}_{0,cyl}^{\infty}(\Z_{- 1})$. As this kind of function is dense in $ \mathscr{C}_{0,cyl}^{\infty}(I \times \Z_{- 1})$, one concludes that $t\in I\to \mu_{t}$  is a solution of the Liouville equation \eqref{2_Liou} for any function $\phi \in \mathscr{C}_{0,cyl}^{\infty}(I \times \Z_{- 1})$.

\end{proof}

\section{Convergence}
\label{2_conv.sec}

In this section, we review the convergence results  proved in \cite{MR3661404}. This part relates the quantum dynamics of many-boson systems as $N\to\infty$  to a characteristic evolution equation and it is based on the analysis of Wigner measures of time-evolved quantum states.  In fact, consider a sequence  $\lbrace \rho_{N} \rbrace_{N \in \mathbb{N}}$ of normal states on $\bigvee^{N} \Z$ with the time evolution,
\[
\rho_{N}(t) := e^{-it H_{N}} \rho_{N}  e^{it H_{N}} \,, \qquad  \widetilde{ \rho}_{N}(t) := e^{itH_{N}^{0}} \rho_{N}(t) e^{-itH_{N}^{0}} ~,
\]
and assume that
\[
\mathcal{M}( \rho_{N} ~,~ N \in \N) = \lbrace \mu_{0} \rbrace ~.~
\]
We show here that for any subsequence of $\lbrace \rho_{N} \rbrace_{N \in \mathbb{N}}$ there exists an extraction $\psi$ such that for all times,
$$
\mathcal{M}( \widetilde\rho_{\psi(N)} ~,~ N \in \N) = \lbrace \widetilde\mu_{t}\rbrace ~.~
$$
 Moreover, the measures $\lbrace\widetilde\mu_t\rbrace_{t\in \R}$ satisfy the characteristic equation introduced in  Section \ref{2_Probrep.sec} with the vector field $v$ given in \eqref{2_int.v}.
Such result on  $\lbrace\widetilde\mu_t\rbrace_{t\in \R}$ is obtained  through the analysis of the quantities of the form
\begin{equation}
\label{2_Tauq}
\mathcal{I}_{N}(t) =\Tr[\widetilde\rho_{N}(t) \,\mathcal{W}(\sqrt{2N} \pi \xi)]\,,
\end{equation}
 according to the following scheme. Firstly, one shows that the quantities $\mathcal{I}_{N}(t)$ satisfy a Duhamel formula given in Proposition \eqref{2_Duhamel}. Secondly, using a standard extraction argument one can show that for any extraction  $\varphi$ one can find a sub-extraction $\psi$ such that for all times $\mathcal{I}_{\psi(N)}(t)$ converges  to a characteristic function of a Borel probability measure $\widetilde\mu_t$ (see Proposition \ref{2_mu boule}). Using such sub-extraction $\psi$ and taking the limit $N\to\infty$ in the Duhamel formula of Proposition \eqref{2_Duhamel}, one obtains a characteristic equation satisfied by the Wigner measures $\widetilde\mu_t$. The key argument employed in the latter convergence is recalled in Proposition \ref{2_keycv}. Finally, the limit equation satisfied by the Wigner measures is identified with the characteristic equation in Section \ref{2_Probrep.sec} related to the vector field $v$ given in \eqref{2_int.v}  (see Lemma  \ref{2_equiv charac}).

We consider the operator $A$ and the quadratic form $q$ of Section \ref{2_prmres.sec} and assume that
the assumptions \eqref{2_A1}-\eqref{2_A2} are satisfied.
For $z \in \Z_{1}$ and $s \in \R$ we define the monomial $q_{s}$:
\[
q_{s}(z) := \frac{1}{2} q( (e^{-isA}z)^{\otimes 2} ,(e^{-isA}z)^{\otimes 2}) \,.
\]
A simple computation yields, for any $ \xi \in \Z_{1}$ and $\varepsilon > 0$:
\[
q_{s}(z+i\varepsilon \pi \xi)-q_{s}(z) = \sum_{j=1}^{4} \varepsilon^{j-1} q_{j}(\xi,s) ~,~
\]
with the monomials $(q_{j}(\xi,s)[z])_{j=1,2,3,4}$ defined by:
\[
q_{1}(\xi,s)[z] = -2\pi Im ~ q(z_{s}^{\otimes 2}, P_{2}\, \xi_{s} \otimes z_{s})  ,\quad  q_{2}(\xi,s)[z] = -\frac{\pi^{2}}{2} Re ~ q (z_{s}^{\otimes 2},\xi_{s}^{\otimes 2}) + 2 \pi^{2} q(P_{2}\, \xi_{s} \otimes z_{s} , P_{2}\, \xi_{s} \otimes z_{s}) ,~
\]

\[
q_{3}(\xi,s)[z] = \pi^{3} Im ~ q(\xi_{s}^{\otimes 2},P_{2} \xi_{s} \otimes z_{s} ) ~ ~ , \qquad~ ~ q_{4}(\xi,s)[z] = \frac{\pi^{4}}{4} q( \xi_{s}^{\otimes 2}, \xi_{s}^{\otimes 2}) ~,~
 \]
where $\xi_{s}$ and $z_{s}$ are the time evolution of $\xi$ and $z$ given as,
\[
\xi_{s} = e^{-isA} \xi ~,~ \qquad z_{s} = e^{-isA} z \,.
\]
According to the $\varepsilon$-dependence of these monomials, one understands easily why the most important term is $q_{1}$ while taking the mean-field limit $\varepsilon=\frac 1 N \to 0$ since $q_{1}$ is the first order in the Taylor expansion.

The following result is proved using similar notations in \cite[Proposition 5.1]{MR3661404}.  It is the key argument behind the convergence of the mean-field approximation and the only place where one needs the properties of relative compactness
\eqref{2_D1} or  \eqref{2_D2}.

\begin{prop}
\label{2_keycv}

Let $\lbrace \rho_{N} \rbrace_{N \in \mathbb{N}}$ be a sequence of normal states on $\bigvee^{N} \Z$  such that $\mathcal{M}( \rho_{N} ~,~ N \in \mathbb{N}) = \lbrace \mu \rbrace$ and suppose that :
\begin{equation}
\label{2_bounda}
\exists C > 0 ~,~ \forall N \in \mathbb{N} ~,~\Tr[ \rho_{N}\, H_{N}^{0} ] \leq CN \,.
\end{equation}
Assume \eqref{2_A1}-\eqref{2_A2} and suppose either \eqref{2_D1} or \eqref{2_D2} is true. Then, for any $\xi \in \Z_1$ and for every $s \in \mathbb{R}$,
\begin{equation}
\label{2_rhoN}
     \lim\limits_{N \to + \infty} \Tr[ \rho_{N}\, \mathcal{W}(\sqrt{2N}\pi \xi) \,[q_{1}(\xi,s)]^{Wick} ] = \int_{\Z} e^{2i\pi Re(\xi,z)} q_{1}(\xi,s)[z] \,d\mu(z),
\end{equation}
where $q_{1}(\xi,s)[z] = -2\pi Im q(z_{s}^{\otimes 2}, P_{2} \xi_{s} \otimes z_{s})$, $z_{s} = e^{-isA}z$ and $\xi_{s}= e^{-isA}\xi$.
\end{prop}

Here we have used the Wick quantization of the monomial $q_{1}$. This operator, $[q_{1}(\xi,s)]^{Wick}$, is taking sense in the Fock space $\Gamma_s(\Z)$. For the standard definition of Wick operators and some of their properties, we refer the reader to \cite[Definition 2.1]{MR2465733} and to \cite[Appendix A]{MR3661404}.

\medskip
We state now a result satisfied by the time dependent quantities $\mathcal{I}_{N}(t) = \Tr[  \widetilde{ \rho}_{N}(t) \mathcal{W}(\sqrt{2N}\pi \xi) ]$. It corresponds to a sort of Duhamel formula or a first order Taylor expansion of $\mathcal{I}_{N}(t)$ with respect to the time variable.

\begin{prop}
\label{2_Duhamel}
We conserve the previous assumptions and notations.
Let $N \in \N$, then for any $\xi \in \Z_{1}$ the map $t \in \mathbb{R} \to \mathcal{I}_{N}(t)$ is of class  $\mathscr{C}^{1}$ and satisfies for all  $t \in \mathbb{R}$,
\begin{equation}
\label{2_IN}
    \mathcal{I}_{N}(t) = \mathcal{I}_{N}(0) + i \int_{0}^{t} \Tr\bigg[ \widetilde{ \rho}_{N}(t) \,\mathcal{W}(\sqrt{2N}\pi \xi) \; [\sum_{j=1}^{4} \varepsilon^{j-1} (q_{j}(\xi,s))^{Wick}]\bigg] \,ds\,,
\end{equation}
where the $q_{j}(\xi,s)$ , $j=1,..,4$ are the monomials given above.
\end{prop}

In order to take the limit $N\to\infty$ in \eqref{2_IN},  one needs to extract subsequences that converge for all times. This is done using a standard diagonal extraction argument in \cite{MR2802894,MR3379490}  and \cite[Proposition 5.2]{MR3661404}. Here we recall such result and  we highlight  the important properties of the Wigner measures.
\begin{prop}\label{2_mu boule}
We conserve the previous assumptions and notations.  For any subsequence  of  $\lbrace \rho_{N} \rbrace_{N \in \mathbb{N}}$
there exists an extraction $\psi$ and a curve $t \in \R \to \widetilde\mu_{t}$ of Borel probability measures in $\mathfrak{P}(\Z)$ such that:
\[
\forall t \in \R ~,~  \mathcal{M}( \rho_{\psi(N)}(t) ~,~ N \in \mathbb{N}) = \lbrace\widetilde \mu_{t} \rbrace\,.
\]
Moreover, for every $t \in \R$ :
\begin{itemize}
\item $\widetilde\mu_{t}$ is supported on the unit closed ball $B_{\Z}(0,1)$ of $\Z$.
\item  $\widetilde\mu_{t}$ is supported on $\Z_{1}$.
\item  $\int_{\Z} ||z||_{\Z_{1}}^{2} \,d\widetilde\mu_{t}(z) \leq C$ with $C$ the constant in \eqref{2_bounda}.
\end{itemize}
\end{prop}

\begin{proof}
The extraction argument follows the one in \cite[Proposition 3.3]{MR2802894} and  \cite[Proposition 5.2]{MR3661404}.
We prove first that for every $t \in \R$, $\widetilde\mu_{t}$ is supported on the closed ball $B_{\Z}(0,1)$. Let $K \in \mathscr{L}^{\infty}(\Z)$ be a compact operator, then according to the previous work of \cite[Theorem 6.13 and Corollary 6.14]{MR2465733}
(see also \cite[Proposition A.5]{MR3661404}) one gets, for $m \in \N$:
\begin{equation}
\label{2_cvwick}
\int_{\Z} \langle z,Az \rangle^{m} \,d\widetilde\mu_{t}(z) = \lim\limits_{k \to \infty} \frac{1}{N_k^m} \,\Tr[\widetilde\rho_{N_{k}}(t)\,d\Gamma(K)^{m}] ~,~
\end{equation}
where $d\Gamma(K)$ is the standard  second quantization of the operator $K$. The Hilbert space $\Z$ is separable so we can find an increasing sequence of compact operators $(K_{n})_{n \in \N}$ such that $K_{n} \to \Id_{\Z}$ strongly. Then, for $n \in \N$ :
\[
\int_{\Z} \langle z,K_{n}z \rangle^{m} \,d\widetilde\mu_{t}(z) = \lim\limits_{k \to \infty}
\frac{1}{N_k^m} \,\Tr[\rho_{N_{k}}(t) \,d\Gamma(K_{n})^{m}] \leq \lim\limits_{k \to \infty} \frac{1}{N_k^m}\,\Tr[\rho_{N_{k}}(t)\,d\Gamma(Id_{\Z})^{m}] = 1  \,,
\]
since, if $B$ and $C$ are two operators such that $B \leq C$, then $d\Gamma(B) \leq d\Gamma(C)$. In the other hand, we obtain:
\[
\lim\limits_{n \to \infty} \int_{Z}  \langle z,A_{n}z \rangle^{m} \,d\widetilde\mu_{t}(z) =  \int_{\Z} \langle z,z \rangle^{m} \,
d\widetilde\mu_{t}(z) = \int_{\Z} ||z||_{\Z}^{2m} \,d\widetilde\mu_{t}(z)  \,.
\]
Finally, we obtain the following bound,
\[
\forall m \in \N ~,~ \int_{\Z} ||z||_{\Z}^{2m} \,d\widetilde\mu_{t}(z) \leq 1\,,\
\]
which subsequently leads  to the equality:
\[
\widetilde\mu_{t}(B_{\Z}(0,1))=1\,.
  \]
The assumption \eqref{2_bounda} on the states $\rho_{N}$ can be extended to all times such that:
\[
 \forall t \in \R ~,~ \forall k \in \N ~,~ \Tr[ \rho_{N_{k}}(t) \,H_{N_{k}}^{0} ] \leq C N_{k}\,.
 \]
Then,  we have:
\[
\forall t \in \R ~,~ \forall k \in \N ~,~ \Tr[ \rho_{N_{k}}(t) \,d\Gamma(A) ] \leq C N_{k}  \,.
\]
Let $B \in \mathscr{L}^{\infty}(\Z)$, $B \leq A$, then a simple estimate yields:
\[
\forall t \in \R ~,~ \forall k \in \N ~,~ \Tr[ \rho_{N_{k}}(t) \,d\Gamma(B) ] \leq C N_{k}\,.
\]
And then using \eqref{2_cvwick}  and taking the limit $k \to\infty$ leads to the inequality:
\[
\forall t \in \R ~,~ \int_{\Z} \langle z,Bz \rangle\, d\widetilde\mu_{t}(z) \leq C  \,.
\]
As this inequality holds true for any compact operator $B$ bounded by $A$, we will use a truncation argument. Let $\chi \in \mathscr{C}_{0}^{\infty}(\R)$ such that $0 \leq \chi \leq 1$, $\chi \equiv 1 $ on $[-1;1]$ and $\chi \equiv 0 $ on $[-2;2]^{c}$. Then for $R >0$, $\chi(\frac{A}{R})A$ is a bounded operator bounded by $A$. Now, it is not difficult to see that a bounded  non-negative  operator can always be weakly approximated  from below by a  sequence of non-negative compact operators $B_n(R)\leq\chi(\frac{A}{R})A\leq A $. So that:
\[
\int_{\Z} \langle z,\chi(\frac{A}{R})Az \rangle \,d\widetilde\mu_{t}(z)=\lim_{n\to\infty} \int_{\Z} \langle z, B_n(R) z \rangle \, d\widetilde\mu_{t}(z)
 \leq C \,,
\]
by dominated convergence. For $z \in \Z$, the following limits hold true:
\[
\lim\limits_{R \to \infty} \langle z,\chi(\frac{A}{R})Az \rangle = \langle z,Az \rangle ~\text{ if } ~ z \in Q(A)=\Z_{1}\,,
  \]
\[
\lim\limits_{R \to \infty} \langle z,\chi(\frac{A}{R})Az \rangle = \infty ~\text{ if } ~ z \notin Q(A)=\Z_{1}  \,.
\]
Using Fatou's lemma, we get:
\[
\int_{\Z} \liminf_{R \to + \infty}  \langle z,\chi(\frac{A}{R})Az \rangle \,d\widetilde\mu_{t}(z) \leq   \liminf_{R \to + \infty}  \int_{\Z}  \langle z,\chi(\frac{A}{R})Az \rangle \,d\widetilde\mu_{t}(z) \leq C\,.
 \]
Hence $\mu_{t}$ is supported on $\Z_1=Q(A)$ for all $t \in \R$  with
\[
\int_{\Z} ||z||_{\Z_{1}}^{2}\,d\widetilde\mu_{t}(z) \leq C \,.
\]
\end{proof}

We recall now the result of convergence proved in \cite[Proposition 5.2]{MR3661404} which is obtained as a combinations
of the Duhamel formula of Proposition \ref{2_Duhamel}  with the convergence argument  of Proposition
 \ref{2_keycv} and the extraction argument  of Proposition  \ref{2_mu boule}.

\begin{prop}\label{2_rho charac}
Let $\lbrace \rho_{N} \rbrace_{N \in \mathbb{N}}$ be a sequence of normal states on $\bigvee^{N} \Z$ satisfying the same  assumptions as in Proposition \ref{2_keycv}.
Then for any subsequence of $\lbrace \rho_{N} \rbrace_{N \in \mathbb{N}}$ there exist an extraction $\psi$ and  a family of Borel probability measures $ \lbrace \widetilde\mu_{t} \rbrace_{t \in \mathbb{R}}$ on $\Z$ such that for all $t \in \mathbb{R}$,
\[
 \mathcal{M}(\widetilde{\rho}_{\psi(N)}(t) ~,~ N \in \mathbb{N} ) = \lbrace \widetilde\mu_{t} \rbrace ~,~
 \]
with  the following characteristic equation  satisfied for any $\xi \in \Z_1=Q(A)$,
\begin{equation}
\label{2_eqchar}
      \widetilde \mu_{t}(e^{2i\pi Re \langle \xi,z \rangle })= \mu_{0}(e^{2i\pi Re \langle \xi,z \rangle}) + i \int_{0}^{t} \widetilde\mu_{s}(e^{2i\pi Re \langle \xi,z \rangle }q_{1}(\xi,s)[z]) ds\,,
\end{equation}
where $q_{1}(\xi,s)[z] = -2\pi Im ~ q(z_{s}^{\otimes 2}, P_{2} \xi_{s} \otimes z_{s})$.
\end{prop}

\begin{lm}\label{2_equiv charac}
We conserve the previous assumptions and notations.  Then the measures $\{\widetilde\mu_t\}_{t\in\R}$ obtained in the
Proposition \ref{2_rho charac} satisfy  the characteristic equation \eqref{2_eqchar2},
\[ \widetilde\mu_{t}(e^{2i\pi Re \langle y,. \rangle_{\Z}}) = \mu_{0}(e^{2i\pi Re \langle y,.\rangle_{\Z}}) + 2i\pi \int_{0}^{t} \widetilde\mu_{s} (e^{2i\pi Re \langle y,. \rangle_{\Z}} Re \langle v(s,.),y \rangle_{\Z}) \,ds\,,
 \]
where  $v$ is the vector field $v(t,z):= -ie^{itA} \partial_{\overline{z}}q_{0}(e^{-itA}z)$ of the mean field equation \eqref{2_int.IVP}.
\end{lm}

\begin{proof}
Remember that the vector field $v(t,z):= -ie^{itA} \partial_{\overline{z}}q_{0}(e^{-itA}z)$ is continuous, bounded on bounded sets
and satisfies the estimate \eqref{2_int.bndv}. So,  the characteristic equation above makes sense. Now, it is  enough to prove that
$$
q_{1}(\xi,s)[z] = 2\pi Re  \langle v(s,z),\xi \rangle_{\Z} \,.
$$
Without losing generality, we prove such assertion for $t=0$.  Let $z \in \Z_{1}$ and $\xi \in \Z_{1}$, then:
\begin{eqnarray*}
&&Re( \langle v(0,z),\xi \rangle_{\Z}) = Re ( \langle \frac{-i}{2} [q(z\otimes . , z^{\otimes 2}) + q(.\otimes z , z^{\otimes 2} ) ], \xi \rangle_{\Z} ) \\
&&= \frac{1}{2} Re ( \langle -i q(z\otimes . , z^{\otimes 2}) + (q(.\otimes z , z^{\otimes 2}), \xi \rangle_{\Z} ) = \frac{1}{2} Re (-i q(z\otimes \xi, z^{\otimes 2})) + \frac{1}{2} Re (-i q(\xi \otimes z, z^{\otimes 2})) \\
&& = \frac{-1}{4} [i q(z\otimes \xi, z^{\otimes 2}) - i \overline{q(z\otimes \xi, z^{\otimes 2})} + i q(\xi \otimes z, z^{\otimes 2}) - i \overline{q(\xi \otimes z, z^{\otimes 2})} ]   \\
&&=  \frac{-i}{2} (q(P_{2} (z \otimes \xi), z^{\otimes 2}) - \overline{q(P_{2} (z \otimes \xi), z^{\otimes 2})}) = Im ( q(P_{2} z\otimes \xi, z^{\otimes 2}) ) \\
&& = -Im ( q(z^{\otimes 2}, P_{2} z \otimes \xi) )\,,
\end{eqnarray*}
where we used the sesquilinearity of the quadratic form $q$.
\end{proof}

\section{Probabilistic representation and uniqueness}
\label{2_Probrep.sec}

In this section, we show that the Liouville equation admits a unique solution
if the related  initial value problem \eqref{2_int.IVP}  verifies the uniqueness property of Definition \ref{2_defweaksol}. Such result improves the one of \cite{MR3721874} and has independent interest. The argument of uniqueness is inspired in one hand by a probabilistic representation for the solutions of the Liouville equation proved in finite dimension in \cite{MR2129498} and extended to infinite dimensional spaces in \cite{MR3379490,MR3721874}; and in the other hand in a  construction of a generalized flow for the initial value problem \eqref{2_int.IVP}  inspired by  \cite{MR2668627} and our recent work \cite{Ammari:2018aa}. So, such uniqueness argument  presented here combined with the convergence results of Section \ref{2_conv.sec}, allow finally to prove our main Theorem \ref{2_main.thm}.

\bigskip
We introduce some useful notations. Let $I$ be an open bounded time interval. We define the space
$$
\mathcal{X} = \Z_{-1} \times \mathscr{C}(\overline{I},\Z_{- 1}),
$$
endowed with the norm:
\[
||(x,\phi)||_{\mathcal{X}} = ||x||_{\Z_{-1}} + \sup_{t \in \overline{I}} ||\phi(t)||_{\Z_{-1}} \,.
\]
For $t \in I$ we define the evaluation map $e_{t}$ over $\mathcal{X}$ as,
\[
e_{t} : (x,\phi) \in \mathcal{X} \to \phi(t) \in \Z_{- 1} \,.
\]
We recall the result proved in \cite[Proposition 4.1]{MR3721874} which justifies the existence of a probability measure $\eta$ concentrated on the solutions of a given initial value problem.

\begin{prop}\label{2_mesureeta}
Consider $v: \R \times \Z_{1} \to \Z_{- 1}$ a Borel vector field such that $v$ is bounded on bounded sets.
Let $I$ be an open bounded time interval containing the origin.
Let $ t \in I \to \mu_{t} \in \mathfrak{P}(\Z_1)$ be a weakly narrowly continuous curve in $\mathfrak{P}(\Z_{- 1})$ satisfying the bound \eqref{2_condliou} and the Liouville equation \eqref{2_Liou} on $I$. Then there exists $\eta$ a Borel probability measure on the space $(\mathcal{X},||.||_{\mathcal{X}})$ satisfying:
\begin{itemize}

\item[(i)] $\eta$ is concentrated on the set of points $(x,\gamma) \in \mathcal{X}$ such that  $\gamma \in W^{1,1}(I,\Z_{- 1})$, $\gamma$ are solutions of the initial value problem \eqref{2_int.IVP} for almost every $t \in I$ and $\gamma(t) \in \Z_{1}$ for almost every $t \in I$ with $\gamma(0)=x_0 \in \Z_{1}$.

\item[(ii)] $ \mu_{t} = (e_{t})\sharp \eta$ for any $t \in I$.
\end{itemize}
\end{prop}

\begin{remark}
\label{2_remv}
According to the Lemma \ref{2_Controle de v}, the vector field $v$ of the mean-field equation \eqref{2_int.IVP} given in
\eqref{2_int.v} satisfies the requested conditions of the above proposition. Moreover, as we will see later on, the family of measures $\widetilde\mu_{t}$ provided by Proposition \ref{2_mu boule} also verifies
the above hypothesis.
\end{remark}

The existence of such a measure $\eta$ is pretty important for us as it concentrates on trajectories of the initial value problem \eqref{2_int.IVP}. The most important implication is the existence of a flow for the initial value problem \eqref{2_int.IVP} if  the uniqueness property of Definition \ref{2_defweaksol} is satisfied.

\medskip
We introduce the set:
\[
\mathfrak{L}^{2,\infty}(\overline{I},\Z_{1}) = \lbrace u \in \mathscr{C}(\overline{I}, \Z_{- 1}) :  ||u(\cdot)||_{L^2(\bar I,\Z_{1})} + ||u(\cdot)||_{L^{\infty}(\overline{I},\Z)} < + \infty \rbrace \,.
\]
We recall a useful property of weakly continuous curves in $\Z_{- 1}$.
\begin{remark}\label{2_rmkcontinuitycurve}
\begin{itemize}
\item Let $\gamma \in \mathscr{C}(\overline{I},\Z_{- 1})$ such that $ ||\gamma(\cdot)||_{L^{\infty}(\overline{I},\Z)} < + \infty$, then $\gamma(t)$ belongs to $\Z$ for every $t \in \overline{I}$ and the function $\gamma:t\in\bar I \to \Z$ is weakly continuous, i.e.~$\gamma \in \mathscr{C}_{w}(\overline{I},\Z)$.

\item As a consequence of this, a curve $u \in \mathfrak{L}^{2,\infty}(\overline{I},\Z_{1})$ is always a $\Z$-valued weakly continuous function on $\bar I$.
\end{itemize}
\end{remark}

\begin{lm}
\label{2_Fmes}
Assume the same assumptions as in Proposition \ref{2_mesureeta}  and consider $v$ to be the vector field  given by \eqref{2_int.v}. Assume further  \eqref{2_A1} and \eqref{2_A2} then
\[
\mathcal{F} := \big\lbrace (x,\gamma) \in \Z_{1} \times \mathfrak{L}^{2,\infty}(\overline{I},\Z_{1}) ~;~ \forall t \in \overline{I} ~,~ \gamma(t) = x + \int_{0}^{t} v(s,\gamma(s))ds \big\rbrace   \,,
\]
is a Borel subset of $\mathcal{X}$ satisfying $\eta(\mathcal{F})=1$, where $\eta$ is the previous probability measure of
Proposition \ref{2_mesureeta}.
\end{lm}

\begin{proof}

For $u \in \mathscr{C}(\overline{I}, \Z_{- 1})$ and $n \in \N^{*}$, we set:
\[
\varphi_{n}(u) = ( \int_{\overline{I}} ||(A+1)(1+\frac{A}{n})^{- 1})u(t)||_{\Z_{-1}}^{2}dt)^{\frac{1}{2}} + \sup_{t \in \overline{I}} ||(A+1)^{\frac{1}{2}}(1+\frac{A}{n})^{\frac{1}{2}}u(t)||_{\Z_{- 1}}\,.
\]
Using simple estimates one can prove that there exist constants $C_{1,n} > 0$ and $C_{2,n} > 0$ depending on $n$ such that:
\[
 \varphi_{n}(u) \leq C_{1,n} ||u||_{L^\infty(\bar I,\Z_{- 1})} ~,~
 \]
\[
\varphi_{n}(u) \geq C_{2,n} ||u||_{L^\infty(\bar I,\Z_{- 1})}  \,.
\]
So, $ \varphi_{n}(\cdot) $ is  an equivalent norm for $ ||\cdot||_{L^\infty(\bar I,\Z_{- 1})}$. Hence,  $ \varphi_{n}$ are continuous. Moreover, for  any fixed $u$, we have the following convergence:
\begin{eqnarray*}
\varphi_{n}(u) \to \varphi(u) &=&  ||u||_{L^2(\bar I,\Z_{1})} + ||u||_{L^{\infty}(\overline{I},\Z_{0})} ~\quad \text{ if }~ u \in \mathfrak{L}^{2,\infty}(\overline{I},\Z_{1})\,,\\
  \varphi_{n}(u) \to \varphi(u) &= &+ \infty ~\quad \text{ if } ~   u \notin \mathfrak{L}^{2,\infty}(\overline{I},\Z_{1})   \,.
\end{eqnarray*}
And so $\varphi$ is a measurable function on $ \mathscr{C}(\overline{I}, \Z_{- 1})$ as a pointwise limit of continuous functions. Hence, we obtain that $\Z_{1} \times\mathfrak{L}^{2,\infty}(\overline{I},\Z_{1})$ is a Borel  subset of $\mathcal{X}=\Z_{- 1} \times \mathscr{C}(\overline{I}, \Z_{- 1})$.

\medskip
\noindent
Let $\psi$ be the function defined as follows
\[
\begin{array}{lrcl}
 \psi :  &  \Z_{1} \times \mathfrak{L}^{2,\infty}(\overline{I},\Z_{1}) & \longrightarrow & \R \\
            &  (x,u) & \longmapsto     & \sup_{t \in \overline{I}} ||u(t) - x - \int_{0}^{t} v(\tau,u(\tau)) d\tau||_{\Z_{-1}}\,.
\end{array}
 \]
Our aim is to prove that $\psi$ is measurable, so that $\mathcal{F} = \psi^{ - 1}(\lbrace 0 \rbrace)$ is a Borel subset of $\mathcal{X}$. For that consider for $t\in\bar I$ the functions,
\[
\begin{array}{lrcl}
 \psi_1^t :  &  \Z_{1} \times \mathfrak{L}^{2,\infty}(\overline{I},\Z_{1}) & \longrightarrow & \Z_{- 1} \\
            &  (x,u) & \longrightarrow     &  u(t)-x\,,
\end{array}
\]
and
\[
\begin{array}{lrcl}
 \psi_{2}^t :  & \mathfrak{L}^{2,\infty}(\overline{I},\Z_{1}) & \longrightarrow & \Z_{- 1} \\
            &  u & \longmapsto     & \int_{0}^{t} v(s,u(s))ds\,.
\end{array}
\]
Note that the function $\psi_1^t$ is continuous and hence it is measurable. In order to prove that
$\psi_{2}^t$ is also measurable, we consider for $t \in \overline{I}$ and $ \varphi \in \Z_{- 1}$, the functions:
\[
\begin{array}{lrcl}
 \psi_{2}^{\varphi,t} :  & \mathfrak{L}^{2,\infty}(\overline{I},\Z_{1}) & \longrightarrow & \R \\
            &  u & \longmapsto     & \int_{t_{0}}^{t} Re \langle v(s,u(s)), \varphi \rangle_{\Z_{- 1}} ds\,.
\end{array}
\]
We show that $\psi_{2}^{\varphi,t}$ is measurable by applying Lemma \ref{2_measurabilityfunction}. Indeed, we take for $[a,b]$ the re-ordered interval $[0,t]$, for  $(M,d)$ the metric space $( \mathfrak{L}^{2,\infty}(\overline{I},\Z_{1}),||.||_{L^\infty(\bar I,\Z_{-1})})$ and for $f$ the function defined for any $(s,u) \in [0,t] \times  \mathfrak{L}^{2,\infty}(\overline{I},\Z_{1})$ by
\[
f(s,u) = Re \langle v(s,u(s)), \varphi \rangle_{\Z_{- 1}}   \,.
\]
Using Lemma \ref{2_Controle de v}, one shows that $f$ is actually continuous as a composition of the following functions,
\[
\begin{array}{ccccccc}
\R\times \Z_1 &\longrightarrow &  \R\times\Z_1   &\longrightarrow & \R\times \Z_{-1}  &\longrightarrow &
\Z_{-1}\\
(s,u) &\longrightarrow & (s,e^{-isA} u)  &\longrightarrow & (s,\partial_{\bar z}q_0(e^{-isA} u))
&\longrightarrow & v(s,u)\,.
\end{array}
\]
Moreover, using again Lemma \ref{2_Controle de v} one checks for any $u \in  \mathfrak{L}^{2,\infty}(\overline{I},\Z_{1})$,
\begin{eqnarray*}
\int_{0}^{t} |f(s,u)|ds &\leq& ||\varphi||_{\Z_{- 1}} \int_{0}^{t} ||v(s,u(s))||_{\Z_{- 1}}ds
\\
&\leq& C \,||\varphi||_{\Z_{- 1}} \int_{0}^{t} (||u(s)||_{\Z_{1}}^{2}||u(s)||_{\Z}^{2}+1)\,ds  \,,
\end{eqnarray*}
 with the latter integral  bounded by $||u||_{L^{\infty}(\overline{I},\Z)}^{2} \int_{\overline{I}} ||u(s)||_{\Z_{1}}^{2} + \lambda(\overline{I}) $. Hence, we can apply  Lemma \ref{2_measurabilityfunction} which ensures that $\psi_{2}^{\varphi,t}$ is a measurable map. Consequently, the map
\[
\begin{array}{lrcl}
 \psi_2^{t} :  & \mathfrak{L}^{2,\infty}(\overline{I},\Z_{1}) & \longrightarrow & \Z_{-1} \\
            &  u & \longmapsto     & \int_{0}^{t}  v(s,u(s)) ds	

\end{array}
\]
is weakly measurable for every fixed $t \in \overline{I}$. The Pettis theorem leads to the measurability of $ \psi_{2}^{t} $ since $\Z_{- 1}$ is a separable Hilbert space (see for instance \cite{MR1336382}). We may now combine the different results to conclude that for every $t \in \overline{I}$ fixed, the function
\[
\begin{array}{lrcl}
 \psi^{t} :  & \Z_{1} \times \mathfrak{L}^{2,\infty}(\overline{I},\Z_{1}) & \longrightarrow & \R \\
            &  (x,u) & \longmapsto     & ||u(t) - x - \int_{t_{0}}^{t}  v(s,u(s)) ds||_{\Z_{-1}}

\end{array}
\]
is measurable. So that
\[
\begin{array}{lrcl}
 \psi :  & \Z_{1} \times \mathfrak{L}^{2,\infty}(\overline{I},\Z_{1}) & \longrightarrow & \R \\
            &  (x,u) & \longmapsto     & \sup_{t \in \mathbb{Q} \cap \overline{I}} ||u(t) - x - \int_{0}^{t}  v(s,u(s)) ds||_{\Z_{-1}}\,,

\end{array}
\]
is also measurable from $\Z_{1} \times \mathfrak{L}^{2,\infty}(\overline{I},\Z_{1})$ to $\R$ since $\mathbb{Q}$ is a countable set. But, as $u \in \mathscr{C}(\overline{I},\Z_{- 1})$ and the map $t \to \int_{t_{0}}^{t} v(s,u(s)) ds \in \mathscr{C}(\overline{I},\Z_{- 1})$, the supremum runs over the whole time interval $\overline{I}$. Finally, we conclude that $ \mathcal{F}$ is a measurable subset of $\mathcal{X} = \Z_{-1} \times \mathscr{C}(\overline{I},\Z_{- 1})$.

\medskip
 According to the first point in the Proposition \ref{2_mesureeta}, there exists a negligible Borel subset $\mathcal{N}$ of $\mathcal{X}$ such that $ \mathcal{N}^{c} \subset \mathcal{F}$. Since $ \mathcal{F}$ is measurable, one  obtains that  $\eta( \mathcal{F}) = 1$.

\end{proof}

\begin{lm}\label{2_gto}

Assume the same assumptions as in Proposition \ref{2_mesureeta}  and consider $v$ to be the vector field  given by \eqref{2_int.v}. Furthermore, suppose  that  the initial value problem \eqref{2_int.IVP} satisfies the uniqueness property of weak solutions as stated in
Definition \eqref{2_defweaksol}.  Then the set
\[
\mathcal{G} := \lbrace x \in \Z_{1} ~,~ \exists \gamma \in \mathfrak{L}^{2,\infty}(\overline{I},\Z_{1}) ~s.t~ (x,\gamma) \in \mathcal{F} \rbrace  \,.
\]
is a Borel subset of $\Z_{1}$ satisfying $\mu_0(\mathcal{G})=1$.

\end{lm}

\begin{proof}

We use the Lemma \ref{2_parthasarathy} given in Appendix \ref{2_appA}  with the metric spaces $X_{1} = (\mathcal{X},||.||_{\mathcal{X}})$, $X_{2} = (\Z_{- 1} , ||.||_{\Z_{- 1}})$ which are complete respectively to their distances and separable and we set:

\[ \begin{array}{lrcl}
 \varphi :  & \mathcal{X} & \longrightarrow & \Z_{- 1} \\
            &  (x,u) & \longmapsto     & x.

\end{array} \]

\noindent
We consider $E_{1} = \mathcal{F}$ and $E_{2} = \mathcal{G}$. Clearly, the map $\varphi$ is continuous and the Lemma \ref{2_Fmes} ensures that the restriction $\varphi|_{\mathcal{F}} : \mathcal{F} \to \Z_{- 1}$ is a measurable map. Suppose now that there exist $(x_{1},u_{1}) \in \mathcal{F}$ and $(x_{2},u_{2}) \in \mathcal{F}$ such that $\varphi|_{\mathcal{F}}(x_{1},u_{1}) = \varphi|_{\mathcal{F}}(x_{2},u_{2})$. Then necessarily, we have $x_{1} = x_{2}$. Because otherwise, $u_1$ and $u_2$ would be two distinct  weak solutions of the initial value problem \eqref{2_int.IVP} with the same initial condition  $x_{1}=x_{2}$.  This shows that $\varphi|_{\mathcal{F}}$ is one-to-one and so $\mathcal{G}= \varphi|_{\mathcal{F}}(\mathcal{F})$ is a Borel subset of $X_{2} = \Z_{- 1}$. But, with the definition of  $\mathcal{G}$ we have that $\mathcal{G} \subset \Z_{1}$, so it is also a Borel subset of $\Z_{1}$.
Using Proposition \ref{2_mesureeta} (ii), one concludes 
$$
\mu_0(\mathcal{G})=(e_0)_\sharp\eta(\mathcal{G})=\eta(e_0^{-1}(\mathcal{G}))= \eta(\mathcal{F})=1\,.
$$

\end{proof}

\begin{lm}
\label{2_genflow}

Assume the same assumptions as in Proposition \ref{2_mesureeta}  and consider $v$ to be the vector field  given by \eqref{2_int.v}.
Suppose furthermore   that  the initial value problem \eqref{2_int.IVP} satisfies the uniqueness property of weak solutions according to
Definition \ref{2_defweaksol}.   Then for any initial condition $x_0\in \mathcal{G}$ there exists a unique weak solution $u(\cdot)$ to the initial value problem \eqref{2_int.IVP}. Moreover, for any time $t\in I$ the  maps

\[ \begin{array}{lrcl}
 \phi(t) :  & \mathcal{G} & \longrightarrow & \Z \\
            &  x_0 & \longmapsto     & u(t),

\end{array} \]

\noindent
where $u(\cdot)$ are such unique weak solutions  are well defined and measurable.

\end{lm}

\begin{proof}

The function $\varphi|_{\mathcal{F}}$ introduced in the proof of Lemma \ref{2_gto} is one-ton-one. So, it is a bijection  from $\mathcal{F}$ to $\mathcal{G}$ and its inverse map

\[ \begin{array}{lrcl}
 \varphi^{- 1} :  & \mathcal{G} & \longrightarrow & \mathcal{F} \\
                  &  x                  & \longmapsto     & (x,u),

\end{array} \]

\noindent
is well-defined and measurable thanks again to Lemma \ref{2_parthasarathy}. And so, the following composition

\[
\begin{array}{cccccc}
 \phi(t): & \mathcal{G}  & \overset{\varphi^{-1}}{\longrightarrow} & \mathcal{F}
  &\overset{e_t}{\longrightarrow} & \Z_{-1}\\
                &       x    & \longrightarrow  & (x,u) & \longrightarrow  & u(t)
\end{array}
\]

\noindent
is well-defined  and measurable since $\varphi^{- 1}$  and $e_{t}$ are measurable maps. Note that $\phi(t)(\mathcal{G}) \subset \Z$ due to the Remark \ref{2_rmkcontinuitycurve}  and remember that $\Z$ is a Borel subset of $\Z_{- 1}$, then the mapping $\phi(t) : \mathcal{G}\to \Z$ is  measurable.
\end{proof}

\medskip
We suppose for the sequel that the assumptions \eqref{2_A1} and \eqref{2_A2} are satisfied.
\begin{cor}

Suppose that the initial value problem \eqref{2_int.IVP} with the vector field $v$ in \eqref{2_int.v} satisfies the uniqueness property of Definition \eqref{2_defweaksol}. Let $ t \in \R \to \mu_{t} \in \mathfrak{P}(\Z_1)$ be a weakly narrowly continuous curve in $\mathfrak{P}(\Z_{- 1})$ satisfying the bound \eqref{2_condliou} and the Liouville equation \eqref{2_Liou}, with the vector field $v$ given in \eqref{2_int.v}, on any open bounded interval  $I$ containing the origin. Then there exist a Borel set 
$\mathcal{G}_0\subset\Z_1$ satisfying $\mu_0(\mathcal{G}_0)=1$ and such that for any initial condition $x_0\in \mathcal{G}_0$ there exists a unique global weak solution $u(\cdot)$ to the initial value problem \eqref{2_int.IVP}. Moreover, for any time $t\in\R$ the  maps

\[ \begin{array}{lrcl}
 \phi(t) :  & \mathcal{G}_0 & \longrightarrow & \Z \\
            &  x_0 & \longmapsto     & u(t),

\end{array} \]

\noindent
 are well defined and measurable.
\end{cor}

\begin{proof}
Note that the set $\mathcal{G}$ defined in Lemma \ref{2_gto} depends on the interval $I$. So, taking an increasing  sequence of intervals $I_n=(-n,n)$ and defining $\mathcal{G}_0=\bigcap_{n\in \N} \mathcal{G}_{I_n}$, one
concludes that $\mathcal{G}_0$ is a Borel set satisfying  
$$
\mu_0(\mathcal{G}_0)=\mu_0(\bigcap_{n\in \N} \mathcal{G}_{I_n})=1\,.
$$
The measurability of the maps $\phi(t)$ is a consequence of Lemma \ref{2_genflow}.
\end{proof}

The above arguments  of measure theory allows to prove  the important property  of uniqueness  for the Liouville equations.

\begin{prop}\label{2_unicityLiouvile}

Consider the vector field $v$ defined in \eqref{2_int.v} and suppose that the related initial value problem \eqref{2_int.IVP} satisfies the uniqueness  property  of weak solutions  according to Definition \ref{2_defweaksol}. If  $t\in \R \to \mu_{t} \in \mathfrak{P}(\Z)$ and $t\in \R \to \nu_{t} \in \mathfrak{P}(\Z)$ are two weakly narrowly continuous curves on $\mathfrak{P}(\Z_{-1})$ satisfying
\eqref{2_condliou} and  the Liouville equation  \eqref{2_Liou} with $\mu_0=\nu_0$, then $\mu_t=\nu_t$ for all times $t\in \R$.

\end{prop}

\begin{proof}
According to the Remark \ref{2_remv}, the vector field $v$ and the measures $\mu_t,\nu_t$ satisfy the hypothesis of Proposition \ref{2_mesureeta}. Hence, there exist  two probability measures $\eta_{\mu}$ and $\eta_{\nu}$ on the space $\mathcal{X}$ satisfying respectively (i)-(ii) of Proposition \ref{2_mesureeta}.  Then for any bounded Borel function $f : \Z_{-1} \to \R$, we have:

\[ \int_{\Z_{-1}} f(x) d\mu_{t}(x) = \int_{\mathcal{X}} f(u(t)) d\eta_{\mu}(x,u) = \int_{\mathcal{F}} f(u(t)) d\eta_{\mu}(x,u)   \]

\[ = \int_{\mathcal{F}} f(\phi(t)(x)) d\eta_{\mu}(x,u) = \int_{\mathcal{G}} f \circ \phi(t)(x) d\mu_{0}(x)\,. \]

\noindent
The first and last equality arise from the  point (ii) in Proposition \ref{2_mesureeta} since  $(e_{t}) \sharp \eta_{\mu} = \mu_{t}$ and
$(e_{0}) \sharp \eta_{\mu} = \mu_{0}$. The second equality arises from the concentration property $\eta_{\mu}(\mathcal{F})=1$.
The third one arises from the existence of the measureable flow  $\phi(t)$ constructed in Lemma \ref{2_genflow}. Finally, we have that for any $t \in I$,

\[ \mu_{t} = \phi(t) \sharp \mu_{0} =  \phi(t) \sharp \nu_{0} = \nu_{t} ~,~ \]

\noindent
since the same arguments above hold true for the measures $\nu_{t}$. Hence, the two measures $\mu_t$ and $\nu_t$ are equal for all times and the Liouville equation \eqref{2_Liou} satisfies the uniqueness property.

\end{proof}

\begin{remark}

The maps $t\mapsto\phi(t)$ may be seen as a "generalized measurable flow" for the initial value problem \eqref{2_int.IVP}.
Since for each $x\in\mathcal{G}\subset \Z_1$ the function $t\mapsto \phi(t)(x)$  is the unique weak solution of the initial value problem \eqref{2_int.IVP} satisfying the initial condition  $ \phi(0)(x)=x$.

\end{remark}

\bigskip
\noindent
\textbf{Proof of the main Theorem \ref{2_main.thm}:}
We can now prove our main theorem.

\begin{proof}
Let $(\rho_{N})_{N \in \mathbb{N}}$ be a sequence of normal states on $\bigvee^{N} \Z$ as in Theorem \ref{2_main.thm}.
Using Proposition \ref{2_rho charac}, one concludes that for any
subsequence of $(\rho_{N})_{N \in \mathbb{N}}$  one can find an extraction $\psi$ and a family of Borel probability measures $\{\widetilde\mu_t\}_{t\in \R}$ on $\Z$ satisfying the  characteristic equation \eqref{2_eqchar}  and  such that for all $t\in\R$,

\[ \mathcal{M}(\widetilde{\rho}_{\psi(N)}(t) ~,~ N \in \mathbb{N} ) = \lbrace \widetilde\mu_{t} \rbrace ~.~ \]

\noindent
We check that the curve of measures $\lbrace \widetilde\mu_{t} \rbrace_{t \in \R}$ satisfies the assumption of Proposition \ref{2_Liouville and characteristic equation}. In fact, thanks to Lemma \ref{2_Controle de v} the vector field $v$ is continuous, bounded on bounded sets and satisfies the estimate \eqref{2_majorv}, i.e.:

\[ \exists C >0 ~,~ \forall t \in \R ~,~ \forall z \in \Z_{1} ~,~ ||v(t,z)||_{\Z_{- 1}}  \leq C (||z||_{\Z_{1}}^{2}.||z||_{\Z}^{2}+1)\,.  \]

\noindent
According to Proposition \ref{2_mu boule}, the measures $\widetilde\mu_{t}$ are supported on $\Z_{1}$ and satisfy for any open bounded time interval $I$ the estimate,

\begin{eqnarray*}
 \int_{I} \int_{\Z_{1}} ||v(t,x)||_{\Z_{- 1}} d\mu_{t}(x) dt &\leq& C \int_{I} \int_{\Z_{1}} (||x||_{\Z_{1}}^{2}||x||_{\Z}^{2}+1)  d\mu_{t}(x) dt
 \\
 & \leq&  C  \int_{I} \int_{\Z_{1}}  ( ||x||_{\Z_{1}}^{2}+1)  d\mu_{t}(x) dt  < + \infty ~,~
 \end{eqnarray*}

 \noindent
since $\lbrace\widetilde\mu_{t} \rbrace_{t \in \R}$  are also  supported  on the unit ball of $\Z$. Adding to this the fact that the set of measures  $\lbrace\widetilde \mu_{t} \rbrace_{t \in \R}$ is weakly narrowly continuous, then one can apply Lemma  \ref{2_equiv charac} and
Proposition  \ref{2_Liouville and characteristic equation} and  conclude that $\lbrace\widetilde\mu_{t} \rbrace_{t \in \R}$ verify actually the Liouville equation \eqref{2_Liou} with the vector field $v$ in \eqref{2_int.v}. Applying the uniqueness property  in Proposition \ref{2_unicityLiouvile} yields that $\lbrace\widetilde \mu_{t} \rbrace_{t \in I}$ is the unique solution of the Liouville  equation \eqref{2_Liou} and
\begin{equation}
\label{2_pushfor}
\widetilde \mu_{t}=\phi(t)_{\sharp} \mu_0\,,
\end{equation}
where $\phi(t)$ is the generalized flow given by Lemma \ref{2_genflow}.

\medskip
To finish the proof note that the above argument leading to \eqref{2_pushfor},   shows actually that  for any time $t\in\R$,
\begin{equation}
\label{2_uniqwig}
\mathcal{M}(\widetilde\rho_{N}(t),~N \in \mathbb{N})=\{\phi(t)_{\sharp} \mu_0\}\,.
\end{equation}
Indeed, for a given time $t\in\R$ if $\nu_t$ is a Wigner measure of $\{\widetilde\rho_{N}(t)\}_{N\in\N}$ then there exists a subsequence which depends in the time $t$  of $\{\widetilde\rho_{N_k}(t)\}_{k\in\N}$ such  that $\mathcal{M}(\widetilde\rho_{N_k}(t),~k\in \mathbb{N})=\{\nu_t\}$. Using the above argument, one deduces the existence of an extraction $\psi$ such that
$$
\mathcal{M}(\widetilde\rho_{\psi(N)}(t),~N\in \mathbb{N})=\{\phi(t)_{\sharp} \mu_0\}\,.
$$
Hence, $\nu_t=\phi(t)_{\sharp} \mu_0$ and this  proves   \eqref{2_uniqwig}. Finally, using the simple relation between $\rho_{N}(t)$ and $\widetilde\rho_{N}(t)$, one deduces all the claimed statements of Theorem \ref{2_main.thm}.

\end{proof}

\vspace{1in}

\begin{center}
\textbf{Acknowledgments}
\end{center}

I would like to thank Zied Ammari for his continued support during the redaction of this paper and for the confidence that he placed in me. The precision, pithiness, and clarity of his reasonings are wonderful qualities which amaze me as well as his abstraction and generalization capacities.

\begin{appendix}
\begin{center}
{\bf Appendix}
\end{center}
\section{Measurability arguments}
\label{2_appA}

We give some useful lemmas about measurability and Borel properties used previously. The Lemma \ref{2_measurabilityfunction} is an adaptation of the Lemma C.2 given in \cite{Ammari:2018aa} and the second one can be found in \cite[Theorem 3.9]{MR0226684},  but we recall it here for the reader's convenience.

\begin{lm}\label{2_measurabilityfunction}

Let $(M,d)$ be a metric space, let $(a,b) \in \R^{2}$, $a<b$. Then, for any measurable function $f:[a,b] \times M \to \R$ such that $\forall u \in M ~,~ f(.,u) \in L^{1}([a,b])$, the mapping $f_{int}$ given by,

\[
\begin{array}{lrcl}
 f_{int} :  &  M & \longrightarrow & \R \\
            &  x & \longmapsto     & \int_{a}^{b}f(s,x)ds

\end{array}
\]

\noindent
is measurable.

\end{lm}

\begin{proof}

Consider the set $\mathcal{F}$:

\[ \mathcal{F}= \lbrace f:[a,b] \times M \to \R ~ measurable ~, ~ \forall u \in M ~,~ f(.,u) \in L^{1}([a,b]) ~ and ~ f_{int} ~ is ~ measurable \rbrace\,. \]

\noindent
We want to prove that $\mathcal{F}$ contians all the  measurable functions such that for all  $ u \in M ~,~ f(.,u) \in L^{1}([a,b]) $. For that we will use the monotone class theorem for functions, see for instance \cite{MR2722836}, Theorem 6.1.2, p.~276.

Let $\mathcal{A}$ be the set of all closed subsets of $[a,b] \times M$, i.e $\mathcal{A} = \lbrace F \subset [a,b] \times M ~,~ F ~ closed \rbrace$. We need $\mathcal{A}$ to be a $\pi$-system to use the monotone class theorem. Recalling that a $\pi$-system is a set stable with respect to intersection, we clearly have that $\mathcal{A}$ is a $\pi$-system. We now check the arguments of the monotone class theorem:

\begin{itemize}

\item $\mathcal{F}$ is stable with respect to addition and scalar multiplication since $(f+\lambda g)_{int} = f_{int} + \lambda g_{int}$ due to the integral linearity.

\item Let $(f_{n})_{n \in \N} \in \mathcal{F}^{\N}$ be a sequence of non-negative functions in $\mathcal{F}$ that increase to a bounded function $f:[a,b] \times M \to \R$. We prove that $f \in \mathcal{F}$.

As the sequence of function is increasing and converge pointwise to the bounded function $f$, we obtain that:

\[ \exists C > 0 ~,~ \forall n \in \N ~,~ \forall (s,u) \in [a,b] \times M ~ |f_{n}(s,u)| \leq f(s,u) \leq C \]

So, for all $u \in M$, $f(.,u) \in L^{1}([a,b])$, the function $f$ is measurable as a limit of measurable functions and the dominated convergence theorem ensure that:

\[ \forall u \in M ~,~ \lim\limits_{n \to + \infty} f_{n,int}(u) = f_{int}(u) \]

So $f_{int}$ is measurable as a limit of measurable functions.

\item One now have to check that if $A \in \mathcal{A}$, then $ \mathbf{1}_{A} \in \mathcal{F}$. Let $A \in \mathcal{A}$, then $ \mathbf{1}_{A}$ is measuable since $A$ is a borelian 	and the Fubini-Tonelli theorem implies that the application $u \in M \to \int_{a}^{b}  \mathbf{1}_{A}(s,u) ds$ is measurable. Moreover, as $[a,b]$ is bounded and $ \mathbf{1}_{A}$ too, we have from the measurability of $ \mathbf{1}_{A}$ that $ \forall u \in M ~,~ \mathbf{1}_{A}(.,u) \in L^{1}([a,b])$ and then : $ \mathbf{1}_{A} \in \mathcal{F}$.

Finally, all the hypothesis of the monotone class theorem for functions are satisfied so $\mathcal{F}$ contains all bounded functions that are measurable with respect to the sigma algebra $\sigma(\mathcal{A})= \mathcal{B}([a,b] \times M)$.

\end{itemize}

Let $f$ be a measurable function $f:[a,b] \times M \to \R$ such that $\forall u \in M ~,~ f(.,u) \in L^{1}([a,b])$. Then, setting for $n \in \N$ the function:

\[ \forall (s,x) \in [a,b] \times M ~,~ f^{n}(s,x) = f(s,x) 1_{ \lbrace |f(s,x)| \leq n \rbrace } (s,x) \]

\noindent
We have that for every $n \in \N$ the function $f_{n}$ is measurable and  bounded. With what precedes, we obtain that :

\[ \forall n \in \N ~,~ f_{n} \in \mathcal{F} ~.~ \]

And so $f_{n,int}$ is measurable for all $n \in \N$. Applying the dominated convergence theorem, since for all $x \in M$ $|f_{n}(.,x)| \leq |f(.,x)| \in L^{1}([a,b])$, we obtain that :

\[ \lim\limits_{n \to + \infty} f_{n,int}(x) = \int_{a}^{b} \lim\limits_{n \to + \infty} f_{n}(s,x) ds = \int_{a}^{b} f(s,x)ds = f_{int}(x) \]

\noindent
and $f_{int}$ is a measurable map as a limit of measurable functions.

\end{proof}

The following general result of measure theory is non trivial and it is very useful in the proof of uniqueness of solutions for the Liouville equation of Section \ref{2_Probrep.sec}. A proof is given in \cite[Theorem 3.9]{MR0226684}.

\begin{lm}
\label{2_parthasarathy}

Let $X_{1}$, $X_{2}$ be two complete separable metric spaces and $E_{1} \subset X_{1}$, $E_{2} \subset X_{2}$ two sets, $E_{1}$ being a Borel set. Let $\varphi$ be a measurable one-to-one map of $E_{1}$ into $X_{2}$ such that $\varphi(E_{1})=E_{2}$.
Then $E_{2}$ is a Borel set of $X_{2}$.

\end{lm}

\section{Vector field bound}
\label{2_appB}
A useful bound is given below  on the vector field $v:\R\times \Z_{1} \to \Z_{-1}$  of the initial value problem  \eqref{2_int.IVP}.  Recall that $v$ is  defined by \eqref{2_int.v}.

\begin{lm}
\label{2_Controle de v}
Assume \eqref{2_A1}-\eqref{2_A2}. Then the vector field $v: \R \times \Z_{1} \to \Z_{- 1}$ given by \eqref{2_int.v} is a continuous mapping bounded on bounded sets of $\R\times\Z_{1}$  and satisfying:
\begin{equation}
\label{2_majorv}
    \begin{aligned}
    && \exists C > 0 ~,~ \forall (t,z) \in \R \times \Z_{1} ~,~ ||v(t,z)||_{\Z_{- 1}} \leq C \, (||z||_{\Z_{1}}^{2}.||z||_{\Z}^{2}+1)\,.
    \end{aligned}
\end{equation}
\end{lm}

\begin{proof}
We show that if $t$ is fixed, then the vector field $v(t,.)$ is continuous from $\Z_{1}$ to $\Z_{- 1}$. As the operator $e^{-itA}$ is unitary, we may suppose without loss of generality that $t=0$. A short computation shows that:

\[ \partial_{\overline{z}} q_{0}(z) = \frac{1}{2} ( q( z \otimes ., z^{\otimes 2}) +
 q( . \otimes z, z^{\otimes 2}))\,.
\]
The application $ \partial_{\overline{z}} q_{0}(z): u \in Q(A) \to \mathbb{C}$ is anti-linear and satisfies:

\[ \partial_{\overline{z}} q_{0}(z)[u]= \frac{1}{2} [q( z \otimes u, z^{\otimes 2}) + q( u \otimes z, z^{\otimes 2})]\,, \]

\[ =  \frac{1}{2} [( z \otimes u, \widetilde{q}( z^{\otimes 2})) + ( u \otimes z, \widetilde{q} (z^{\otimes 2}))] \,.
\]
The quadratic form $q$ is sesquilinear. The assumption \eqref{2_A2} yields a certain bound on $q(x,x)$ for $x \in \Z_{1}$. Using the following polarization formulas, one gets that  for any $(x,y) \in \Z_{1}^{2}$:

\[ q(x,y) = \frac{1}{4} \sum_{k=0}^{3} q(y+i^{k}x,y+i^{k}x)i^{k} ~,~ \]

\[ q(x,x)+q(y,y) = \frac{1}{4} \sum_{k=0}^{3} q(x+i^{k}y,x+i^{k}y) \,.
\]
We may now estimate $|q(x,y)|$:
\[
|q(x,y)| = |q(y,x)| \leq  \frac{1}{4} \sum_{k=0}^{3} |q(x+i^{k}y,x+i^{k}y)|\,,
\]

\[
 \leq  \frac{a}{4} \sum_{k=0}^{3} \langle x+i^{k}y,(A_{1}+A_{2})x+i^{k}y \rangle_{\Z^{\otimes 2} } +  \frac{b}{4} \sum_{k=0}^{3} \langle x+i^{k}y,x+i^{k}y \rangle_{\Z^{\otimes 2} } ~,~
   \]
since all the terms in the sum are non-negative.
But using the second polarization formula for a different quadratic form ( $\overline{q}\equiv(A_{1}+A_{2})$ or $\overline{q}\equiv Id_{\Z^{\otimes 2} }$) , one gets:

\[  \frac{a}{4} \sum_{k=0}^{3} \langle x+i^{k}y,(A_{1}+A_{2})x+i^{k}y \rangle_{\Z^{\otimes 2} } = a \big(\langle x,(A_{1}+A_{2})x \rangle_{\Z^{\otimes 2} } + \langle y,(A_{1}+A_{2})y \rangle_{\Z^{\otimes 2} }\big)\,\]
and
\[
 \frac{b}{4} \sum_{k=0}^{3} \langle x+i^{k}y,x+i^{k}y \rangle_{\Z^{\otimes 2} } = b( \langle x,x \rangle_{\Z^{\otimes 2} } + \langle y,y \rangle_{\Z^{\otimes 2} })\,.
 \]
So, we have the following estimate :

\begin{equation}
\label{2_contv}
 |q(x,y)| \leq a (\langle x,(A_{1}+A_{2})x \rangle_{\Z^{\otimes 2} } + \langle y,(A_{1}+A_{2})y \rangle_{\Z^{\otimes 2} }) + b( \langle x,x \rangle_{\Z^{\otimes 2} } + \langle y,y \rangle_{\Z^{\otimes 2} }) \,.
\end{equation}
Now we will use this estimate to obtain the right bound on the vector field $v$. Let $\Phi \in \Z_{1}$ such that $||\Phi||_{\Z_{1}}=1$, let $z \in \Z_{1}$, then taking $x= z \otimes \Phi$ and $ y= z^{\otimes 2}$ one obtains:
\[
|q(x,y)| \leq a ( \langle z \otimes \Phi,(A_{1}+A_{2})z \otimes \Phi \rangle_{\Z^{\otimes 2} } + \langle z^{\otimes 2} ,(A_{1}+A_{2}) z^{\otimes 2} \rangle_{\Z^{\otimes 2} }) + b ( \langle z \otimes \Phi, z \otimes \Phi \rangle_{\Z^{\otimes 2} } +\langle z^{\otimes 2} , z^{\otimes 2} \rangle_{\Z^{\otimes 2} })
  \]

\[
\leq a[ ||z||_{\Z_{1}}^{2} || \Phi ||_{\Z}^{2} + ||\Phi||_{\Z_{1}}^{2} ||z||_{\Z}^{2}+ 2 ||z||_{\Z_{1}}^{2} || z||_{\Z}^{2}] + b [ ||z||_{\Z}^{2}||\Phi||_{\Z}^{2} +  ||z||_{\Z}^{4}] \,.
\]
\medskip
As $||\Phi||_{\Z_{1}} = 1$ and  $||.||_{\Z} \leq ||.||_{\Z_{1}}$ then:
\[
|q(x,y)| \leq a (  ||z||_{\Z_{1}}^{2} + ||z||_{\Z}^{2} + 2||z||_{\Z_{1}}^{2} ||z||_{\Z}^{2}) + b(||z||_{\Z}^{2}+||z||_{\Z}^{4})\,.
  \]
Finally, we can find a constant $C > 0$ such that:
\[
|q(z \otimes \Phi,z^{\otimes 2})| \leq C( ||z||_{\Z_{1}}^{2} ||z||_{\Z}^{2} +1)  \,.
\]
Remember now that $v(0,z) = \frac{1}{2} ( q(z\otimes ., z^{\otimes 2}) +  q(.\otimes z, z^{\otimes 2}))$. So using the sesquilinearity of $q$, one obtains:
\[ |
|v(0,z)||_{\Z_{- 1}} \leq C( ||z||_{\Z_{1}}^{2} ||z||_{\Z}^{2} +1) ~.~
 \]
This proves that the anti-linear mapping  $\partial_{\overline{z}} q_{0}(z):u \to \partial_{\overline{z}} q_{0}(z)[u] $ is continuous and  bounded by $ C( ||z||_{\Z_{1}}^{2} ||z||_{\Z}^{2} +1)$. This in particular  ensures that the mapping  $z \in \Z_{1} \to \partial_{\overline{z}} q_{0}(z) \in \Z_{- 1}$ is  bounded on bounded sets of $\Z_{1}$ and satisfies the estimate \eqref{2_majorv}. The continuity of the vector field 
$v$ follows in a similar way since the difference $v(t,z)-v(s,w)$ can always be controlled with the 
estimate \eqref{2_contv}.

\end{proof}

\end{appendix}

\bibliographystyle{plain}

\end{document}